\documentclass[twoside,12pt]{article}
\usepackage[T1]{fontenc}
\usepackage[utf8]{inputenc}
\usepackage[margin=1in]{geometry}
\usepackage[toc,page]{appendix}

\usepackage{amssymb}
\usepackage{tabularx}
\usepackage{tikz}
\usetikzlibrary{patterns,hobby}
\usepackage{amsmath}
\usepackage{algorithm}
\usepackage{algpseudocode}
\usepackage{mathtools}
\usepackage{enumerate}

\usepackage{url}
\usepackage{hyperref}
\hypersetup{colorlinks,allcolors=[rgb]{0,0.13672,0.95703}}
\usepackage{graphicx}
\usepackage{subfig}
\usepackage{wrapfig}
\usepackage{color,soul}
\usepackage{float}
\usepackage{setspace}
\usepackage{tikz,pgfplots}
\pgfplotsset{compat=newest}
\usepackage{etoolbox}
\usepackage{comment}
\usepackage{natbib}

\usepackage{fancyhdr}
\title{Faster Low-Rank Approximation and Kernel Ridge Regression
via the Block-Nystr\"om Method}
\author{Sachin Garg\thanks{Computer Science and Engineering, University of Michigan, \texttt{sachg@umich.edu}}\and  Micha{\l} Derezi{\'n}ski\thanks{Computer Science and Engineering, University of Michigan, \texttt{derezin@umich.edu}} }
\date{}

\newcommand{\calh}{\mathcal{H}}
\newcommand{\calc}{\mathcal{C}}
\newcommand{\cals}{\mathcal{S}}
\newcommand{\calz}{\mathcal{Z}}
\newcommand{\calp}{\mathcal{P}}
\newcommand{\call}{\mathcal{L}}
\newcommand{\calx}{\mathcal{X}}

\def\nnz{{\mathrm{nnz}}}

\def\qed{\hfill$\blacksquare$\medskip}

\def\W{\mathbf W}

\def\T{\mathbf T}
\def\R{\mathbf R}

\def\K{\mathbf K}

\ifx\BlackBox\undefined
\newcommand{\BlackBox}{\rule{1.5ex}{1.5ex}}  
\fi
\DeclareMathOperator*{\argmin}{\mathop{\mathrm{argmin}}}

\DeclareMathOperator*{\diag}{\mathop{\mathrm{diag}}}
\def\x{\mathbf x}

\def\y{\mathbf y}

\def\a{\mathbf a}
\def\b{\mathbf b}

\def\w{\mathbf w}
\def\v{\mathbf v}

\def\e{\mathbf e}

\def\u{\mathbf u}

\def\ee{\mathrm{e}}

\def\B{\mathbf B}

\def\A{\mathbf A}
\def\C{\mathbf C}

\def\M{\mathbf M}

\def\S{\mathbf S}

\def\Z{\mathbf Z}

\def\I{\mathbf I}

\def\A{\mathbf A}

\def\P{\mathbf P}
\def\Q{\mathbf Q}

\def\E{\mathbb E}

\def\R{\mathbb R} 
\def\tr{\mathrm{tr}}

\def\cS{{\mathcal{S}}}

\let\origtop\top
\renewcommand\top{{\scriptscriptstyle{\origtop}}} 

\definecolor{silver}{cmyk}{0,0,0,0.3}
\definecolor{yellow}{cmyk}{0,0,0.9,0.0}
\definecolor{reddishyellow}{cmyk}{0,0.22,1.0,0.0}
\definecolor{black}{cmyk}{0,0,0.0,1.0}
\definecolor{darkYellow}{cmyk}{0.2,0.4,1.0,0}
\definecolor{darkSilver}{cmyk}{0,0,0,0.1}
\definecolor{grey}{cmyk}{0,0,0,0.5}
\definecolor{darkgreen}{cmyk}{0.6,0,0.8,0}

\newcommand{\Green}[1]{{\color{darkgreen}  {#1}}}
\newcommand{\Blue}[1]{\color{blue}{#1}\color{black}}
\newcommand{\Brown}[1]{{\color{brown}{#1}\color{black}}}

\ifx\proof\undefined
\newenvironment{proof}{\par\noindent{\bf Proof\ }}{\hfill\BlackBox\\[2mm]}
\fi

\ifx\theorem\undefined
\newtheorem{theorem}{Theorem}
\fi

\ifx\example\undefined
\newtheorem{example}{Example}
\fi

\ifx\property\undefined
\newtheorem{property}{Property}
\fi

\ifx\lemma\undefined
\newtheorem{lemma}{Lemma}
\fi

\ifx\proposition\undefined
\newtheorem{proposition}{Proposition}
\fi

\ifx\remark\undefined
\newtheorem{remark}{Remark}
\fi

\ifx\corollary\undefined
\newtheorem{corollary}{Corollary}
\fi

\ifx\definition\undefined
\newtheorem{definition}{Definition}
\fi

\ifx\conjecture\undefined
\newtheorem{conjecture}[theorem]{Conjecture}
\fi

\ifx\axiom\undefined
\newtheorem{axiom}[theorem]{Axiom}
\fi

\ifx\claim\undefined
\newtheorem{claim}[theorem]{Claim}
\fi

\ifx\assumption\undefined
\newtheorem{assumption}{Assumption}
\fi

\def\qbar{\bar{q}}
\def\Bb{\overline{\B}}
\def\Csize{m}
\def\DB{\mathcal{D}_{\text{\small bless}}}
\def\DA{\mathcal{D}_{\text{\small dpp}}}
\def\ie{i.e.,\xspace}
\def\eg{e.g.,\xspace}
\def\whp{w.h.p.,\xspace}
\def\nystrom{Nystr\"om\xspace}
\def\dppvfx{\textsc{DPP-VFX}\xspace}
\newcommand{\wh}[1]{\widehat{#1}}
\newcommand{\sizeS}[1]{|S_{#1}|}

\newcommand\rnumber{\mathop{\mbox{$r$-$\mathit{number}$}}}

\newcommand{\HRule}{\rule{\linewidth}{0.5mm}}
\newcommand{\Hrule}{\rule{\linewidth}{0.3mm}}
\newcommand{\red}[1]{\textcolor{red}{#1}}
\newcommand\es{{\mathcal{E}}}
\newcommand\ep{{\mathcal{E}^{'}}}
\newcommand\eone{{\mathcal{E}_1}}
\newcommand\etwo{{\mathcal{E}_2}}
\newcommand\ethr{{\mathcal{E}_3}}

\newcommand\ind{{\boldsymbol{1}}}
\newcommand{\nsq}[1]{\left\lVert#1\right\rVert^2}
\newcommand{\Qm}{\Q_{-i}}
\newcommand{\Qmj}{\Q_{-ij}}
\newcommand{\Qmr}{\Q_{-ir}}
\newcommand{\Qmjk}{\Q_{-ij_{k}}}

\newcommand{\norm}[1]{\left\lVert#1\right\rVert}

\newcommand{\Answer}[1]{\textcolor{magenta}{#1}}

\begin{document}

\maketitle

\begin{abstract}%
  The Nystr\"om method is a popular low-rank approximation technique for large matrices that arise in kernel methods and convex optimization. Yet, when the data exhibits heavy-tailed spectral decay, the effective dimension of the problem often becomes so large that even the Nystr\"om method may be outside of our computational budget. To address this, we propose Block-Nystr\"om, an algorithm that injects a block-diagonal structure into the Nystr\"om method, thereby significantly reducing its computational cost while recovering strong approximation guarantees. We show that Block-Nystr\"om can be used to construct improved preconditioners for second-order optimization, as well as to efficiently solve kernel ridge regression for statistical learning over Hilbert spaces. Our key technical insight is that, within the same computational budget, combining several smaller Nystr\"om approximations leads to stronger tail estimates of the input spectrum than using one larger approximation. Along the way, we provide a recursive preconditioning scheme for efficiently inverting the Block-Nystr\"om matrix, and provide new statistical learning bounds for a broad class of approximate kernel ridge regression solvers.
\end{abstract}

\section{Introduction}

Fast algorithms for approximating large positive semidefinite (psd) matrices are central to many problems in computational learning. For instance, when minimizing a convex training loss, approximating the psd Hessian matrix is a key step in preconditioning first-order optimization methods \citep[e.g.,][]{qu2016sdna} and designing Newton-type algorithms \citep[e.g.,][]{erdogdu2015convergence}. Moreover, for methods such as Gaussian processes \citep{williams2006gaussian} and kernel ridge regression \citep[KRR, e.g.,][]{bach2013sharp}, approximating the psd kernel matrix is an essential step for making these techniques scalable to large datasets.

The Nystr\"om method has proven to be one of the most successful techniques for psd matrix approximation \citep{williams2000using}, leading to numerous Nystr\"om-preconditioned optimization algorithms \cite[e.g.,][]{avron2017faster,frangella2023randomized,frangella2024promise}, as well as learning algorithms based on Nystr\"om kernel approximations \cite[e.g.,][]{rudi2015less,rudi2018fast,burt2019rates}. In the most basic version of the method, given an $n\times n$ psd matrix $\A$, we select a subset $S\subseteq\{1,...,n\}$ of $m$ coordinates (landmarks), and use them to construct $\hat\A = \C\W^{-1}\C^\top$, where $\C=\A_{:,S}$ is the submatrix of the columns of $\A$ indexed by $S$, whereas $\W=\A_{S,S}$ is the principal submatrix of $\A$ indexed by $S$. Given a subset $S$, this approximate decomposition can be obtained by evaluating $nm$ entries of $\A$ and inverting the principal submatrix $\W$ at the cost of $O(m^3)$ operations. Essentially, this method projects $\A$ onto a subspace defined by the landmarks, thus it is crucial to ensure that the subset $S$ is small but representative of $\A$'s structure.

One of the key advantages of the Nystr\"om method is that it combines good practical performance with strong theoretical guarantees. To ensure the latter, a long line of works \citep[including][]{ridge-leverage-scores,musco2017recursive,rudi2018fast} have developed a randomized landmark selection technique called approximate ridge leverage score sampling. Given some $\lambda>0$, we define $\lambda$-ridge leverage scores of $\A$ as the diagonal entries of $\A(\A+\lambda\I)^{-1}$, and their sum is called the $\lambda$-effective dimension, $d_\lambda(\A) := \tr(\A(\A+\lambda\I)^{-1})\leq n$. It is known that to obtain a strong Nystr\"om approximation of $\A$ along the top part of its spectrum above the $\lambda$ threshold, we must sample $m=\tilde O(d_\lambda(\A))$ landmarks according to its $\lambda$-ridge leverage scores (we use $\tilde O$ to hide logarithmic factors). Formally, this gives the following guarantee for $\hat \A$ as an approximation of $\A$:
\begin{align}
 \text{($\lambda$-regularized $\alpha$-approximation)}\quad   \alpha^{-1}(\A+\lambda\I)\preceq \hat\A+\lambda\I\preceq \A+\lambda\I,\qquad\qquad\label{eq:approx}
\end{align}
where $\preceq$ denotes the Loewner ordering and the above sampling scheme achieves $\alpha=2$. Such $\alpha$-approximations have several downstream applications, including where $\alpha$ is the approximation factor in the statistical risk of Nystr\"om-approximated KRR (Section \ref{s:app-hilbert}), or where $\alpha$ represents the condition number after preconditioning an $\ell_2$-regularized Hessian matrix (Section~\ref{s:app-precond}).

While Nystr\"om approximations are efficient, they still carry a substantial computational cost when the $\lambda$-effective dimension is large (i.e., when the regularizer $\lambda$ is small or the spectrum of $\A$ is heavy-tailed). Obtaining a ridge leverage score sample of size $m=\tilde O(d_\lambda(\A))$ takes up to $\tilde O(nm^2)$ operations \citep{musco2017recursive}, while the cost of applying and/or inverting the matrix $\hat\A+\lambda\I$ requires additional $\tilde O(nm+m^3)$ time. This leads to quadratic or even cubic dependence on the landmark sample size $m$, which is a challenge for learning tasks with large effective dimension, such as KRR in the unattainable setting \citep{lin2020convergences} and other problems with heavy-tailed spectral distributions. This leads to our main motivating question:
\begin{quote}
    \textit{How to construct a $\lambda$-regularized psd matrix approximation when $\lambda$-effective dimension is so large that the classical Nystr\"om approximation is outside of our time budget?}
\end{quote}

To address this question, we propose Block-Nystr\"om, a cost-effective alternative to the Nystr\"om method, which, roughly speaking, sparsifies the Nystr\"om principal submatrix $\W=\A_{S,S}$ into a block-diagonal structure. This not only reduces the inversion cost of $\W$, but also decreases the effective dimension associated with the method, thereby reducing the cost of ridge leverage score sampling. We show that Block-Nystr\"om achieves a better approximation factor $\alpha$ than the classical Nystr\"om under the same computational budget when the effective dimension is large and the spectral decay is heavy-tailed. We summarize Block-Nystr\"om in the following theorem.

\begin{theorem}[Block-Nystr\"om]\label{t:block-nystrom}
    Given a psd matrix $\A\in\R^{n\times n}$, $\lambda>0$, and $\alpha\geq 1$, we can construct a matrix data structure $\hat\A$ of rank $m=\tilde O(\alpha d_{\alpha^2\lambda}(\A))$ in time $\tilde O\big((nm^2+m^3)/\alpha^2 + nm\big)$ such that:
   
    \begin{enumerate}
        \item $\hat\A$ is a $\lambda$-regularized $O(\alpha)$-approximation of $\A$ in the sense of \eqref{eq:approx};
        \item We can apply $\hat\A$ to a vector in $\tilde O(nm)$ time;
        \item We can apply $(\hat\A+\lambda\I)^{-1}$ to a vector in $\tilde O\big(nm\cdot\alpha^{o(1)}\big)$ time.
    \end{enumerate}
   
\end{theorem}

\paragraph{Construction} Let $S = (S_1, S_2, ... , S_q)$ be a collection of $m=qb$ coordinates drawn according to approximate $\alpha^2\lambda$-ridge leverage score sampling, split up into $q=\tilde O(\alpha)$ blocks of size $b = \tilde O(d_{\alpha^2\lambda}(\A))$. We define the Block-Nystr\"om approximation as follows:

\begin{align}
\hat\A=\C\W^{-1}\C^\top\quad\text{for}\quad 
\C=\A_{:,S}\quad\text{and}\quad\W=q\diag\big(\A_{S_1,S_1},...,\A_{S_q,S_q}\big).\label{eq:block-nystrom}
\end{align}

Note that Block-Nystr\"om can be viewed as an extension of the classical Nystr\"om approximation (obtained by setting $q=1$). In fact, Theorem \ref{t:block-nystrom} recovers the standard Nystr\"om guarantee (including all of the time complexities up to logarithmic factors) for $\alpha=1$. In that case, we obtain an $O(1)$-approximation at the cost of $\tilde O(nm^2/\alpha^2)=\tilde O(nd_\lambda(\A)^2)$. As we relax the approximation guarantee by increasing~$\alpha$, the effective dimension decreases to $d_{\alpha^2\lambda}(\A)$, thus reducing this dominant cost. 

Can a similar relaxed approximation guarantee be achieved with the classical Nystr\"om method? Yes, but it will not reduce the computational cost as substantially as Block-Nystr\"om. A simple calculation shows that a Nystr\"om approximation constructed from $\tilde O(d_{\alpha\lambda}(\A))$ landmarks drawn according to $\alpha\lambda$-ridge leverage scores also achieves an $O(\alpha)$-approximation. However, the cost will be higher than Block-Nystr\"om, since $d_{\alpha\lambda}(\A)\geq d_{\alpha^2\lambda}(\A)$, where the gap depends on the spectral profile of $\A$. As a concrete example, suppose that the spectrum of $\A$ exhibits polynomial decay, $\lambda_i(\A) =\Theta(i^{-1/\gamma})$, as is typical in kernel methods. Then, we have $d_{\alpha\lambda}(\A) = \Omega(\alpha^{\gamma} d_{\alpha^2\lambda}(\A))$, so Block-Nystr\"om's $\tilde O(nd_{\alpha^2\lambda}(\A)^2)$ runtime is faster than classical Nystr\"om by a factor of $\tilde\Omega(\alpha^{2\gamma})$, which gets more significant the more heavy-tailed the spectrum of $\A$ is. A detailed cost comparison is deferred to the discussion of the concrete applications of Block-Nystr\"om.

\paragraph{Key technical insights} Why does Block-Nystr\"om offer a computational advantage over classical Nystr\"om? To offer some intuition, we discuss two key technical insights that unlocked our analysis.
\begin{enumerate}
    \item \textit{Optimizing bias-variance trade-off}. The Block-Nystr\"om construction \eqref{eq:block-nystrom} can be equivalently formulated as an average of smaller classical Nystr\"om approximations, $\hat\A_{[q]} = \frac1q\sum_{i=1}^q\hat\A_i$, where $\hat\A_i = \C_i\W_i^{-1}\C_i^\top$ with $\C_i=\A_{:,S_i}$ and $\W_i=\A_{S_i,S_i}$. Thus, $\hat\A_{[q]}$ should behave similarly to the expectation $\E[\hat\A_1]$. While taking expectation over the landmark samples does not substantially improve the approximation of the top part of $\A$'s spectrum (bias dominates variance), it dramatically improves the tail estimates (variance dominates bias). Thus, via careful matrix concentration analysis, we can optimize this bias-variance trade-off to choose the optimal amount of averaging. This also explains why Block-Nystr\"om makes the biggest gains for heavy-tailed spectra (e.g., polynomial decay), where tail estimates matter the most. 
   
    \item \textit{Fast inversion via recursive preconditioning}. In most applications of the Nystr\"om method, quickly computing the regularized inverse $(\hat\A+\lambda\I)^{-1}$ is crucial. This is usually done via the Woodbury formula: $(\C\W^{-1}\C^\top+\lambda\I)^{-1} = \frac1\lambda(\I-\C(\C^\top\C+\lambda\W)^{-1}\C^\top)$, which requires computing and inverting the matrix $(\C^\top\C+\lambda\W)^{-1}$. Unfortunately, adding $\C^\top\C$ eliminates the block-diagonal structure of $\lambda\W$, thereby seemingly erasing the computational benefits of Block-Nystr\"om. We propose an alternative approach, inspired by the recursive preconditioning scheme of \cite{derezinski2025approaching}, where we use an iterative solver to implicitly compute the inverse $(\hat\A_{[q]}+\lambda\I)^{-1}$. We precondition this solver by a smaller Block-Nystr\"om matrix with increased regularization, $\hat\A_{[q/c]}+c^2\lambda\I$ for $c=\tilde O(1)$, and repeat the procedure recursively until reaching the single Nystr\"om matrix that can be inverted via Woodbury. We show that this scheme implements inversion in almost-linear time.
\end{enumerate}
\subsection{Application: Preconditioning convex optimization}
\label{s:app-precond}
One of the primary settings where the Nystr\"om method has received considerable attention is convex optimization. Here, the task is to minimize a convex function $g:\R^n\rightarrow\R$ by utilizing its gradient and Hessian information. Prior works have used Nystr\"om and other psd matrix approximations of the Hessian $\nabla^2g(\x)$ to extract second-order information from the function by constructing a preconditioner for a first-order method, both in the general convex setting \citep{erdogdu2015convergence,ye2020nesterov,frangella2024promise,sun2025sapphire} as well as when $g$ is a quadratic \citep{avron2017faster,frangella2023randomized,dmy24}. In many cases, these preconditioners are treated as a black box, so that one can simply apply our Block-Nystr\"om method in place of existing approaches, and use guarantee \eqref{eq:approx} to ensure fast linear convergence for strongly convex objectives.

To illustrate this, we apply our Theorem \ref{t:block-nystrom} to a convergence result for a Nesterov-accelerated approximate Newton method \citep[Theorem 2 in][]{ye2020nesterov}, which implies that Block-Nystr\"om can be effectively used to achieve fast linear convergence in a neighborhood around the solution.

\begin{corollary}[Block-Nystr\"om preconditioning]\label{c:opt}
    Let $g(\x) = \psi(\x) + \lambda\|\x\|^2$, where $\psi:\R^n\rightarrow\R$ is convex with an $L$-Lipschitz-continuous Hessian, and if $L>0$, suppose that $\x_0$ is in a sufficiently small neighborhood around $\x^*=\argmin_\x g(\x)$. Also, let $\hat\A$ be the Block-Nystr\"om approximation of $\nabla^2\psi(\x_0)$ with regularizer $\lambda$ and approximation factor $\alpha\geq 1$. Then, for any $\epsilon>0$, after $\tilde O(\sqrt\alpha\log1/\epsilon)$ applications of $(\hat\A+\lambda\I)^{-1}$ and of $\nabla\psi$, we can obtain $\hat\x$ such that:
    \begin{align*}
        g(\hat\x)- g(\x^*)\leq \epsilon\cdot\big[g(\x_0)-g(\x^*)\big].
    \end{align*}
\end{corollary}
To provide a concrete computational comparison between Block-Nystr\"om and other preconditioners, 
 in Section \ref{s:quadratic-main} we study the task of minimizing a quadratic function $g(\x) = \frac12\x^\top\A\x - \x^\top\b$ under heavy-tailed spectrum. We illustrate this here on the class of input matrices that follow the spiked covariance model \citep{johnstone2001distribution}, i.e., consisting of non-isotropic low-rank signal distorted by isotropic noise. Following \cite{dy24}, we model this by allowing some top-$k$ part of $\A$'s spectrum (the signal) to be arbitrarily ill-conditioned, while requiring the remaining tail (the noise) to have condition number $O(1)$. We show that a careful implementation of Block-Nystr\"om preconditioning achieves better time complexity for solving this quadratic minimization than the previous best known Nystr\"om-based preconditioner \citep{dmy24} when the number of large eigenvalues is $O(n^{0.82})$, improving the runtime from $\tilde O(n^{2.065})$ to $\tilde O(n^{2.044})$. Similar gains can be obtained for matrices with other heavy-tailed spectral decay profiles.
\begin{corollary}\label{c:quadratic}
    Given a strongly convex quadratic $g(\x) = \frac12\x^\top\A\x - \x^\top\b$ such that $\A\in\R^{n\times n}$ has at most $O(n^{0.82})$ eigenvalues larger than $O(1)$ times its smallest eigenvalue, there is an algorithm that with high probability computes $\hat\x$ such that $g(\hat\x)\leq (1+\epsilon)\cdot \min_{\x\in\R^n}g(\x)$ in time $\tilde O(n^{2.044}\log1/\epsilon)$.
\end{corollary}

\subsection{Application: Least Squares Regression over Hilbert spaces}
\label{s:app-hilbert}
Another significant application of the Nystr\"om method is in speeding up kernel ridge regression for learning over Hilbert spaces. In this model, we receive $n$ i.i.d.~samples $\{(\x_i,y_i)\}_{i=1}^n$ from an unknown probability measure $\rho$ over $\mathcal H\times \R$ where $\mathcal H$ is a Hilbert space, and our goal is to approximate the expected risk minimizer under square loss, $\inf_{f\in\mathcal H}\E_{\rho}[(f(\x)-y)^2]$, where $f(\x):=\langle f,\x\rangle_{\mathcal H}$. When the minimizing hypothesis $f_{\mathcal H}$ lies in $\mathcal H$ (the attainable case), this gives rise to a very well understood regression framework, which can be effectively learned via kernel ridge regression:
\begin{align*}
    \hat f = \sum_{i=1}^n w_i\x_i,\quad\text{where}\quad \w = (\K+n\lambda^*\I)^{-1}\y,\quad \K=\big[\langle\x_i,\x_j\rangle_{\mathcal H}\big]_{ij}.
\end{align*}
Optimal learning rates for kernel ridge regression have been derived in the attainable case \citep{caponnetto2007optimal}, in particular ensuring that the effective dimension under optimal regularizer $n\lambda^*$ is always bounded by $O(\sqrt n)$. This enables using efficient Nystr\"om-based approximations of $\K$ which recover the optimal rate with a low computational cost \citep{rudi2015less,rudi2017falkon,rudi2018fast}.

However, the situation is different in the considerably harder unattainable case \citep{lin2020convergences,lin2020optimal}, where the minimizing hypothesis does not lie in the Hilbert space, but rather, in a larger function space, as controlled by a smoothness parameter $\zeta\in(0,1)$ (see Assumption \ref{as:3}). 
In the unattainable case (i.e., $\zeta<1/2$), the effective dimension is often much larger (even close to~$n$),
making the Nystr\"om method far more expensive if we wish to recover the convergence rate of full KRR. Thus, it is natural to ask how much we can reduce our computational budget by relaxing the desired convergence rate so that we can use a coarser $\alpha$-approximation of $\K$. We do this by considering $n\lambda^{*}$-regularized $\alpha$-approximation of $\K$, where $\alpha=O(n^\theta)$ for small enough $\theta>0$.

Unfortunately, existing risk analysis for KRR applies only to projection-based kernel approximations, such as the classical Nystr\"om method. Block-Nystr\"om is no longer a strict projection, so to apply it, we must first extend the KRR risk analysis in the unattainable case to a broader class of kernel $\alpha$-approximations. We do this in the following result, which is stated here informally for Block-Nystr\"om, but applies much more broadly and thus should be of independent interest.

\begin{theorem}[Block-Nystr\"om KRR, informal Corollary \ref{c:avg_nys_risk}]
    Given Hilbert space $\mathcal H$ and probability measure $\rho$ over $\mathcal H\times \R$ with smoothness $\zeta\in(0,1/2)$ (Assumption~\ref{as:3}) and capacity $\gamma\in(0,1)$ (Assumption~\ref{as:4}), let $\{(\x_i,y_i)\}_{i=1}^n$ be i.i.d. samples from $\rho$. Let $\hat\K$ be the Block-Nystr\"om approximation (Theorem~\ref{t:block-nystrom}) of $\K=[\langle\x_i,\x_j\rangle_{\mathcal H}]_{i,j}$ with the optimal KRR regularizer $n\lambda^*$ and some $\alpha\geq 1$. Then, the estimate $\hat f=\sum_i\hat w_i\x_i$ constructed from $\hat\w=(\hat\K+n\lambda^*\I)^{-1}\y$ achieves risk bound:
    \begin{align*}
        \|\hat f - f_{\mathcal H}\|_\rho\ \lesssim\ \alpha\cdot n^{-\frac{\zeta}{\min\{1,2\zeta+\gamma\}}}.
    \end{align*}
\end{theorem}
\begin{remark}
    The above convergence rate matches the rate achieved by full KRR (i.e., without kernel approximation) obtained by \cite{lin2020optimal}, up to the $\alpha$ factor. The same guarantee can also be obtained for an $\alpha$-approximation  based on the classical Nystr\"om method.
\end{remark}
In Section \ref{s:hilbert}, we use this analysis to evaluate the computational gain of relaxing the convergence rate exponent by setting $\alpha = n^{\theta}$ for $\theta>0$ small enough to still ensure convergence. We show that, similarly as in the general psd matrix approximation task, Block-Nystr\"om achieves a better computational gain than $n\lambda^*$-regularized classical Nystr\"om when $\theta$ is small but positive, across all values of smoothness and capacity in the unattainable setting.

\subsection{Background and further related work}

The Nystr\"om method \citep{nystrom1930praktische} was first used in a machine learning context by \cite{williams2000using} to approximate kernel matrices arising in Gaussian process regression. Initially, uniform sampling was used for landmark selection, with other sampling schemes considered by \cite{kumar2009sampling} among others, until \cite{ridge-leverage-scores} showed that sampling landmarks according to ridge leverage scores leads to strong approximation guarantees in terms of the effective dimension. Learning guarantees for Nystr\"om-based KRR in the fixed design model were considered by \cite{bach2013sharp}, and then adapted to Hilbert spaces by \cite{rudi2015less,rudi2017falkon}. Fast ridge leverage score sampling schemes were designed for this setting by \cite{rudi2018fast}. Recent works have considered extensions of Nystr\"om to the unattainable setting \citep[e.g.,][]{lin2020optimal} and distributed architectures \citep[e.g.,][]{li2023optimal}. Other approximate KRR solvers have also been studied, including random features \citep[e.g.,][]{rudi2017generalization} 
and stochastic optimization \citep[e.g.,][]{lin2018optimal,abedsoltan2023toward}.

The Nystr\"om method has also been studied in randomized linear algebra as a low-rank approximation method \citep[e.g.,][]{halko2011finding}, with many generalizations such as using sketching instead of landmark sampling \citep{gittens2013revisiting}. One of the key applications in this context is preconditioning iterative algorithms for solving quadratic minimization to high precision \citep[e.g.,][]{avron2017faster,frangella2023randomized,diaz2023robust,abedsoltan2024nystrom,dmy24}, as well as for other convex optimization settings \citep[e.g.,][]{frangella2024promise}.

\section{Preliminaries}
 
 Let $\lambda_i$ denote the $i^{th}$ eigenvalue of $\A$, with $\lambda_1/\lambda_n := \kappa$ denoting the condition number of $\A$. We use $\preceq$ to denote the Loewner ordering and define $\|\x\|_{\A}=\sqrt{\x^\top\A\x}$.  $\S\in\R^{s\times n}$ is a sub-sampling matrix for $\{p_i\}_{i=1}^n$ if its rows are standard basis vectors $\e_i^\top$ sampled i.i.d. proportionally to $p_i$'s.
  
\begin{definition}[$\lambda$-ridge leverage scores of $\A$]\label{d:ridge_lev}
    Given a psd matrix $\A$ and $\lambda>0$, we define the $i$th $\lambda$-ridge leverage score $\ell_i(\A,\lambda)$ as $[\A(\A+\lambda\I)^{-1}]_{i,i}$. Also we define $d_\lambda(\A):=\tr(\A(\A+\lambda\I)^{-1})$ as the $\lambda$-effective dimension. We say that $\tilde \ell_1,...,\tilde\ell_n$ are $O(1)$-approximate $\lambda$-ridge leverage scores of $\A$  if $\tilde \ell_i \geq \ell_i(\A,\lambda)/T$ for all $i$ and $T=O(1)$, and $\sum_i\tilde\ell_i = O(d_\lambda(\A))$.
\end{definition}
In our results, we rely on two different algorithms for computing approximate ridge leverage scores.

\begin{lemma}[Adapted from Theorem 7 of \cite{musco2017recursive}]\label{l:rls_musco}
There exists an algorithm that provides $\tilde \ell_i(\A,\lambda)$ i.e., $O(1)$-approximate $\lambda$-ridge leverage scores of $\A$ in time $\tilde O(nd_\lambda(\A)^2)$. 

\end{lemma}
If we additionally assume that $\A_{ii} =O(1)$ for all $i$, which is common in kernel literature, then it is possible to generate a leverage score sample even faster than what is guaranteed in Lemma~\ref{l:rls_musco}. 
\begin{lemma}[Adapted from Theorem 1 of \cite{rudi2018fast}]\label{l:rls_bless}
   Let $\A_{ii} =O(1)$ for all $i$. Given $\lambda>0$, we can generate a sub-sampling matrix $\S \in \R^{s\times n}$ with $s=\tilde O(d_\lambda(\A))$, sampled i.i.d. from $O(1)$-approximate $\lambda$-ridge leverage scores of $\A$ in time $\tilde O\big(\min(\frac{1}{\lambda},n)d_{\lambda}(\A)^2\big)$.

\end{lemma}
Ridge leverage score sampling provides the following standard Nystr\"om approximation guarantee.
\begin{lemma}[Theorem 3 of \cite{musco2017recursive}]\label{l:musco-approx}
   If $\S$ is a sub-sampling matrix corresponding to $\tilde O(d_\lambda(\A))$ i.i.d. samples according to $O(1)$-approximate $\lambda$-ridge leverage scores of $\A$, then the Nystr\"om matrix  $\hat\A = \A\S^\top(\S\A\S^\top)^{-1}\S\A$ with high probability satisfies $\|\hat\A-\A\|\leq \lambda$.
\end{lemma}

\section{Analysis of Block-Nystr\"om}
In this section, we give the analysis of Block-Nystr\"om method. First, we provide the approximation guarantees by optimizing bias-variance trade-offs in Nystr\"om, and then we describe our recursive preconditioning algorithm for inverting the regularized Block-Nystr\"om matrix. We conclude the section with the proof of Theorem \ref{t:block-nystrom}.

\subsection{Optimizing bias-variance trade-offs}

First, let us discuss how the classical Nystr\"om guarantees relate to the $\lambda$-regularized $\alpha$-approximation condition \eqref{eq:approx}. Note that, given an $n\times n$ psd matrix $\A$ and a sub-sampling matrix $\S\in\R^{s\times n}$, the Nystr\"om approximation can be defined as $\hat\A=\A^{1/2}\P\A^{1/2}$, where $\P=\A^{1/2}\S^\top(\S\A\S^\top)^\dagger\S\A^{1/2}$ is an orthogonal projection matrix onto the row-span of $\S\A^{1/2}$. This implies that $\hat\A\preceq \A$. Moreover, if $\S$ is produced by $\lambda$-ridge leverage score sampling of $s=\tilde O(d_\lambda(\A))$ landmarks, then standard guarantees (e.g., Lemma \ref{l:musco-approx}) give $\|\A-\hat\A\|\leq \lambda$. In particular, this implies that $\A+\lambda\I\preceq \hat\A+2\lambda\I\preceq 2(\hat\A+\lambda\I)$. Together, these properties yield \eqref{eq:approx} with $\alpha=2$.

What if we wish to construct a smaller Nystr\"om approximation with rank $\tilde O(d_{\lambda'}(\A))$ for some $\lambda'>\lambda$, and still use it with a small regularizer $\lambda$? In that case, we can trivially observe that $\A+\lambda\I\preceq (1+\frac{\lambda'}{\lambda})(\hat\A+\lambda\I)$, which implies guarantee \eqref{eq:approx} with $\alpha=1+\lambda'/\lambda$. This is not particularly satisfying, and we could not really hope to do much better, since $\hat\A$ only has $\tilde O(d_{\lambda'}(\A))$ non-zero eigenvalues, which are associated primarily with the eigenvalues of $\A$ above the $\lambda'$ threshold. Thus, it contains little information about $\A$'s spectrum in the $[\lambda,\lambda']$ interval.

However, if we could use the expectation $\E[\hat\A]=\A^{1/2}\E[\P]\A^{1/2}$ as the approximation of~$\A$, then the story would change. Even though $\P$ is a low-rank projection, its expectation $\E[\P]$ is the same rank as $\A$, so it carries some information about all directions in $\A$'s spectrum. Expected projection matrices arise in many contexts, and a number of works have studied their properties under various sketching and sub-sampling matrices $\S$ \citep[e.g.,][]{rodomanov2020randomized,randomized-newton,derezinski2024sharp}, including for $\lambda'$-ridge leverage score sampling \cite[][]{askotch}. From these results, we can show that if the sample size is at least  $\tilde O(d_{\lambda'}(\A))$, then:
\begin{align}
    \text{(Lemma \ref{rls-projection})}\qquad\E[\P]\succeq \frac12 \A(\A+\lambda'\I)^{-1}.\label{eq:rls-proj}
\end{align}
This characterization suggests that $\E[\hat\A]$ provides meaningful estimates even for the spectral tail of~$\A$. Indeed, we show that given a Nystr\"om approximation $\hat\A$ based on $\lambda'$-ridge leverage score sampling, for any $\lambda\leq \lambda'$ we have:
\begin{align}
         (\E[\hat\A] +\lambda\I) \succeq \sqrt{\lambda/4\lambda'}\cdot(\A+\lambda\I).
\end{align}
So, $\E[\hat\A]$ is a non-trivial $O(\sqrt{\lambda'/\lambda})$-approximation of $\A$ for any $\lambda<\lambda'$. This is compared to the $O(\lambda'/\lambda)$-approximation achieved by $\hat\A$, which is essentially vacuous in the tail of the spectrum. 

Naturally, $\E[\hat\A]$ is not a practical approximation strategy, but it provides the motivation for our Block-Nystr\"om method. Our approach is to approximate $\E[\hat\A]$ by drawing several independent copies $\hat\A_1$, ..., $\hat\A_q$, and averaging them together, to obtain $\hat\A_{[q]}=\frac1q\sum_{i=1}^q\hat\A_i$.  This reduces the inherent variance of $\hat\A$ in the tail of the spectrum and, with large enough $q$, should lead to an $O(\sqrt{\lambda'/\lambda})$-approximation, matching the guarantee we obtained for $\E[\hat\A]$. 

Fortunately, via a matrix Chernoff argument, we are able to show that a relatively moderate number of Nystr\"om copies is sufficient to recover the improved guarantee, showing that $\hat\A_{[q]}$ is with high probability an $\tilde O(\sqrt{\lambda'/\lambda} + \lambda'/(q\lambda))$-approximation. Here, the first term can be viewed as the contribution of the bias, whereas the second term is the variance which goes down with $q$. 

Optimizing over this bias-variance error decomposition, we conclude that to obtain an $O(\alpha)$-approximation, one can choose $\lambda'=\alpha^2\lambda$, and then set $q = \tilde O(\sqrt{\lambda'/\lambda})=\tilde O(\alpha)$. We summarize this in the following result, which is proven in Appendix \ref{a:concentration}.

\begin{theorem}[Spectral approximation with Block-Nystr\"om]\label{t:Nys_spectral_approximation}
Given $\lambda>0$ and $\alpha \geq 1$, let $\{p_i\}_{i=1}^n$ denote the $O(1)$-approximate $\alpha^2\lambda$-ridge leverage scores of $\A$. Let $\hat\A_{[q]}=\frac{1}{q}\sum_{i=1}^{q}{\hat\A_i}$ where $\hat\A_i$'s are iid Nystr\"om approximations of $\A$ sampled i.i.d.~from $\{p_i\}_{i=1}^n$, each with $O(d_{\alpha^2\lambda}(\A)\log^3(n/\delta))$ landmarks.
    If $q \geq 32\alpha\log(n/\delta)$, then with probability $1-\delta$, the matrix $\hat\A_{[q]}$ is a $\lambda$-regularized $16\alpha$-approximation of $\A$, i.e.,
    \begin{align*}
        (64\alpha)^{-1}(\A+\lambda\I)\preceq\hat\A_{[q]}+\lambda\I\preceq\A+\lambda\I.
    \end{align*}

\end{theorem}

\subsection{Fast inversion via recursive preconditioning}

The computational benefits in the preprocessing cost of Block-Nystr\"om are easy to verify: The cost of constructing a single Nystr\"om approximation $\hat\A=\C\W^{-1}\C^\top$ scales cubically with the number of landmarks, because we must invert the matrix $\W$, but the cost of constructing many small Nystr\"om approximations grows only linearly with the number of landmarks. These benefits become even more significant when we account for the cost of ridge leverage score sampling. However, it is much less straightforward to retain that speed-up if we need to also quickly invert the approximation, as is the case in kernel ridge regression, for example. Remarkably, we are able to show that this can still be done efficiently, by significantly departing from standard block-inversion arguments.

Our task is to quickly apply $(\hat\A_{[q]}+\lambda\I)^{-1}$ to any given vector $\v \in \R^n$. 
A standard approach that can be applied to most low-rank approximations is to use the Woodbury inversion formula, which implies that for any matrix $\hat\A=\C\W^{-1}\C^\top$ we have $(\hat\A+\lambda\I)^{-1} = \frac1\lambda(\I-\C(\C^\top\C+\lambda\W)^{-1}\C^\top)$. Thus, it suffices to pre-compute the matrix $(\C^\top\C+\lambda\W)^{-1}$ in $O(nm^2+m^3)$ time, where $m$ is the number of landmarks, so that we can apply the inverse to a vector $\v$ in $O(nm)$ operations. This can be further improved by using an iterative solver for applying the inverse $(\C^\top\C+\lambda\W)^{-1}$, reducing the preprocessing cost to $\tilde O(nm+m^3)$ \citep{dmy24}. In principle, this strategy could still be applied to invert the Block-Nystr\"om approximation by relying on decomposition \eqref{eq:block-nystrom}. However, unfortunately the matrix $\C^\top\C+\lambda\W$ does not retain the block-diagonal structure of $\W$, so the preprocessing cost of this approach must scale cubically with the number of landmarks, thus erasing most of the computational benefit of Block-Nystr\"om.

We apply a different strategy, called recursive preconditioning, which has been used extensively for solving Laplacian systems (e.g., see \cite{js_talg} and references therein). This approach has also recently been adapted for solving general dense linear systems by \cite{derezinski2025approaching}.
We specialize it for Block-Nystr\"om, and show that one can still almost-match the fast $O(nm)$ inversion time for Block-Nystr\"om without incurring any additional preprocessing cost. We achieve this by recursively reducing the task of applying $(\hat\A_{[q]}+\lambda\I)^{-1}$ to the cheaper task of applying $(\hat\A_1+\tilde\lambda\I)^{-1}$ for a carefully chosen $\tilde\lambda$, where $\hat\A_1$ is one of the small Nystr\"om components in $\hat\A_{[q]}$ (proof in Appendix~\ref{a:concentration}).

\begin{theorem}[Recursive solver for Block-Nystr\"om]\label{t:linear_system_solve}
Given $\lambda >0$, consider some $\alpha \geq 1$ and $q = O \big(\alpha\log(n/\delta)/\theta^2\big)$, and let $\hat\A_{i}$ for $1\leq i \leq q$ be Nystr\"om approximations of $\A$ sampled from $O(1)$-approximate $\alpha^2\lambda$-ridge leverage scores of $\A$, each of rank $ O(d_{\alpha^2\lambda}(\A)\log^3(n/\delta))$. Denote $\hat\A_{[q]} =  \frac{1}{q}\sum_{i=1}^{q}{\hat\A_i}$. Then, conditioned on an event with probability $1-\delta$, for any $\v \in \R^n$, $\epsilon>0$ and $0<\phi<1$ we can find $\u \in \R^n$ such that:
\begin{align*}
    \Big\|\u-\big(\hat\A_{[q]} +\lambda\I\big)^{-1}\v\Big\|
    \leq \epsilon\cdot \|\v\| 
\end{align*}
in time $\tilde O\big(c\cdot \alpha^{1+\phi}nd_{\alpha^2\lambda}(\A)\log1/\epsilon\big)$ with a preprocessing cost of $\tilde O(\alpha(nd_{\alpha^2\lambda}(\A)+d_{\alpha^2\lambda}(\A)^3))$, for any $c$ such that $c \geq \big\lceil\max\big\{\log^{2/\phi}(20c^2(1+\theta)^4),2\big\} \big\rceil.$
\end{theorem}
\begin{remark}
Choosing $\theta=1/2$, $\phi = \frac1{\sqrt{\log c}}$ and sufficiently large $c = O(\log n)$, we get  runtime $\tilde O(\alpha^{1+o(1)}nd_{\alpha^2\lambda}(\A)) = \tilde O(nm\cdot \alpha^{o(1)})$, where $m=\tilde O(\alpha d_{\alpha\lambda}(\A))$ is the total number of landmarks across all $q$ Nystr\"om matrices. Similarly, the preprocessing cost can be written as $\tilde O(nm + m^3/\alpha^2)$.
\end{remark}
\paragraph{Proof Sketch.} For any integer $k\geq 0$, consider $q_k = \tilde O(\frac{\alpha}{\theta^2c^k})$ and $\M_k = \hat\A_{[q_k]} + c^{2k}\lambda\I$. Observe that we need to approximate $\M_0^{-1}\v$. The main idea here is to solve the linear system $\M_0\u=\v$ to $\epsilon$ accuracy by preconditioning $\M_0$ with $\M_1$, which in turn leads to solving the linear system $\M_1\u=\v$ for some $\v \in \R^n$. In fact, by solving the linear system $\M_1\u=\v$ to a sufficient accuracy, one can obtain an $\epsilon$-approximate solution to $\M_0\u=\v$ (see Lemma \ref{l:pcg}). Now recursing on $k$, and letting $k = O\big(\log(\alpha)/\log(c)\big)$, one reaches $\M_{k+1}$ with $q_{k+1} = \tilde O(1).$ At that stage, we precondition $\M_{k+1}$ with just $\hat\A_1 + \lambda\I$. 

The first main argument of the proof justifies preconditioning of $\M_{j}$ with $\M_{j+1}$ by proving that for any $0 \leq j \leq k$, we have $\kappa(\M_{j}^{-1/2}\M_{j+1}\M_{j}^{-1/2})\leq c^2(1+\theta)^4$. We obtain this
by relying on matrix concentration arguments to show that
$\kappa\big(\M_{j}^{-1/2}\E\M_{j}\M_{j}^{-1/2}\big) \leq (1+\theta)^2$ and also observing 
that $\kappa\big((\E\M_{j})^{-1/2}\E\M_{j+1}(\E\M_{j})^{-1/2}\big) \leq c^2$ for each $j$. By design, there is a trade-off associated with choosing $c$, with large $c$ leading to computational gains at the expense of quality of preconditioning and vice-a-versa. Furthermore, as each recursive step leads to approximately solving a linear system, the approximation error can accumulate multiplicatively with the depth of the recursion. Despite all that, we show that one can pick $c$ carefully enough such that the overall runtime complexity is $\tilde O(c\cdot q^{1+\phi}nd_{\alpha^2\lambda}(\A))$ for any $\phi>0$. We finish the proof using a recent result (Lemma~\ref{l:dmy_Nystrom_solve}) from \cite{dmy24} stating that $(\hat\A_1 + \alpha^2\lambda\I)^{-1}\u=\v$ can be solved appproximately in time $\tilde O(nd_{\alpha^2\lambda}(\A))$ after an $\tilde O(nd_{\alpha^2\lambda}(\A)+d_{\alpha^2\lambda}(\A)^3)$ preprocessing cost.

\subsection{Proof of Theorem \ref{t:block-nystrom}}
Let $q=\tilde O(\alpha)$ and $\hat\A_{[q]}+\lambda\I= \frac{1}{q}\sum_{i=1}^{q}{\hat\A_i+\lambda\I}$, where each $\hat\A_i$ is a Nystr\"om approximation of $\A$ of rank $\tilde O(d_{\alpha^2\lambda}(\A))$ sampled from $O(1)$-approximate $\alpha^2\lambda$-ridge leverage scores of $\A$. Let $m=\tilde O(\alpha d_{\alpha^2\lambda}(\A))$ be the number of landmarks across all $\hat\A_i$'s. The ridge leverage scores can be approximated using recursive techniques from \cite{musco2017recursive}, resulting in cost $\tilde O(nd_{\alpha^2\lambda}(\A)^2) = \tilde O(nm/\alpha^2)$. Furthermore, as each individual $\hat\A_i$ has the decomposition $\C_i\W_i^{-1}\C_i^\top$ where $\C_i \in \R^{n \times \tilde O(d_{\alpha^2\lambda}(\A))}$ and $\W_i \in \R^{ \tilde O(d_{\alpha^2\lambda}(\A)) \times  \tilde O(d_{\alpha^2\lambda}(\A))}$, we compute $\C_i$ and $\W_i^{-1}$ for all $q$ Nystr\"om approximations costing $\tilde O(\alpha(nd_{\alpha^2\lambda}(\A) + d_{\alpha^2\lambda}(\A)^3)) = \tilde O(nm + m^3/\alpha^2)$. The total construction cost is therefore $\tilde O\big((nm+m^3)/\alpha^2 + nm\big)$.
In Theorem  \ref{t:Nys_spectral_approximation} we have shown that $\hat\A_{[q]}+\lambda\I$ forms an $O(\alpha)$-approximation of $\A+\lambda\I$ in the sense of (\ref{eq:approx}). Secondly, for any vector $\v \in \R^n$, multiplication of $\hat\A_i$ with $\v$ can be carried out sequentially from right to left by exploiting the decomposition of $\hat\A_i$, resulting in total cost $q\cdot \tilde O(nd_{\alpha^2\lambda}(\A))=\tilde O(nm)$. Furthermore, as shown in Theorem~\ref{t:linear_system_solve}, $\big(\hat\A_{[q]} +\lambda\I\big)^{-1}\v$ can be approximated to high precision in time $\tilde O\big(nm\cdot\alpha^{o(1)}\big)$. 

\section{Block-Nystr\"om Preconditioner for Quadratic\\ Minimization}
\label{s:quadratic-main}

In this section, we present an application of Block-Nystr\"om for minimizing a quadratic function $g(\x) = \frac12\x^\top\A\x-\x^\top\b$, given an $n\times n$ psd matrix $\A$ and a vector $\b\in\R^n$, specifically using Block-Nystr\"om as a preconditioner for minimizing $g(\x)$, concluding with the proof of Corollary~\ref{c:quadratic}. For a more direct comparison with recent related works, we will state the time complexities in this section using the fast matrix multiplication exponent $\omega\approx2.372$.

We consider the setting of Corollary $\ref{c:quadratic}$. Let $\A$ have at most $k$ eigenvalues larger than $O(1)$ times its smallest eigenvalue. Our first step is to approximate $\bar\lambda$-ridge leverage scores of $\A$, where $\bar\lambda=\frac{1}{k}\sum_{i>k}{\lambda_i(\A)}$. The recursive sampling method from \cite{musco2017recursive} provides $O(1)$-approximations for all $\bar\lambda$-ridge leverage scores in time $\tilde O(nk^{\omega-1})$, however, we show that if matrix $\A$ has flat-tailed spectrum, then the cost of leverage score approximation can be further optimized to $\tilde O(\nnz(\A)+k^\omega)$.

\begin{lemma}\label{l:rls-square-main}
    Given an $n\times n$ psd matrix $\A$ with at most $k$ eigenvalues larger than $O(1)$ times its smallest eigenvalue, consider $\bar\lambda=\frac1k\sum_{i>k}\lambda_i(\A)$. Then, $d_{\bar\lambda(\A)}\leq 2k$, and we can compute $O(1)$-approximations of all $\bar\lambda$-ridge leverage scores of $\A$ in time $\tilde O(\nnz(\A) + k^\omega)$.
\end{lemma}
Note that the $\tilde O(\nnz(\A)+k^\omega)$ runtime may be preferrable to the $\tilde O(nk^{\omega-1})$ runtime attained by \cite{musco2017recursive} when $k$ is sufficiently large or the matrix is sufficiently sparse. The proof of Lemma \ref{l:rls-square-main} can be found in Appendix \ref{a:preconditioner}. We are now ready to prove Corollary \ref{c:quadratic}.

\paragraph{Proof of Corollary \ref{c:quadratic}} 
   We invoke the Block-Nystr\"om method (Theorem \ref{t:block-nystrom}) for $\lambda =\bar\lambda/\alpha^2$, obtaining $\hat\A$, using ridge leverage sampling based on Lemma \ref{l:rls-square-main}. Estimating the ridge leverage scores takes time $\tilde O(\nnz(\A)+k^\omega)$. Obtaining the data structure $\hat\A$ requires preprocessing, where we sample each $\hat\A_i$ for $1\leq i \leq \tilde O(\alpha)$ and compute $\hat\C_i$, $\W_i^{-1}$. This preprocessing cost adds a factor of $\tilde O(\alpha(nk+k^\omega))$ to the overall cost. Now, we use $(\hat\A+\lambda\I)^{-1}$ as the preconditioner for an iterative quadratic solver such as the one given in Corollary \ref{c:opt}. Note that Theorem \ref{t:block-nystrom} implies that $\hat\A$ is a $\lambda$-regularized $O(\alpha)$-approximation of $\A$. Therefore, condition number after preconditioning $\A$ with $\hat\A+\lambda\I$ will then be
    \begin{align*}    \kappa\Big((\hat\A+\lambda\I)^{-1/2}\A(\hat\A+\lambda\I)^{-1/2}\Big) 
    &=  O\bigg(\alpha\Big(1+\frac\lambda{\lambda_{\min}}\Big)\bigg) 
    = O\Big(\alpha + \frac{n}{k\alpha}\Big),     
    \end{align*}
    where in the last step we used that $\bar\lambda = O(\frac nk\lambda_{\min})$ by the assumption on $\A$. Setting $\alpha=\sqrt{n/k}$, we obtain that the condition number after preconditioning is $O(\sqrt{n/k})$. Every iteration of the preconditioned solver requires matrix vector products with $\A$ costing $O(n^2)$ and with $(\hat\A+\lambda\I)^{-1}$, costing $\tilde O(\alpha^{1+o(1)}nk)$ , due to Lemma \ref{t:linear_system_solve}.  Furthermore, for $\alpha = O(\sqrt{n/k})$ we have $\alpha nk = n\sqrt{nk}\leq n^2$. Thus, the overall cost of the solver includes $\tilde O(n^2+k^\omega\alpha)$ for preprocessing, and then $\tilde O((n^2+\alpha^{1+o(1)}nk)\sqrt\alpha)$ for solving. Note that again we have $\alpha^{1+o(1)}nk=O(n^2)$. So, overall, we get the time complexity:\vspace{-4mm}
    \begin{align*}
        \tilde O\bigg(n^2\Big(\frac nk\Big)^{1/4} + k^\omega\Big(\frac nk\Big)^{1/2}\bigg).
    \end{align*}
    Now, observe that if the first term dominates the second term, then we can drive $k$ up to balance them. So, by optimizing over $k$, we get $k=n^{\frac{2-1/4}{\omega-1/4}}\approx n^{0.82}$, and plugging this into the bound we get $\tilde O(n^{\frac{\omega-2}{4\omega-1}}+k^\omega\sqrt{n/k}) = \tilde O(n^{2.044})$ for any $k\leq n^{0.82}$.\qed
\paragraph{Discussion:} The Block-Nystr\"om preconditioner achieves $\tilde O(n^{2.044} + k^\omega\sqrt{n/k})$ time complexity, which is better than the $\tilde O(n^{0.065}+k^\omega)$ runtime of the preconditioner by \cite{dmy24} for any $k\leq n^{0.82}$. Interestingly, an altogether different approach to this problem was proposed by \cite{dy24}. They used a stochastic solver instead of a preconditioner, which means that their method cannot be paired with deterministic solvers and exploit fast matrix vector products with $\A$, unlike preconditioner-based methods such as Block-Nystr\"om). Their approach achieves $\tilde O(n^2+nk^{\omega-1})$ runtime, so our method also improves upon this approach, for any $k\geq n^{0.76}$.

\section{Least Squares Regression over Hilbert Spaces}
\label{s:hilbert}
In this section, we present an application of Block-Nystr\"om to kernel ridge regression over Hilbert spaces in hard learning scenarios. Along the way, we give a new risk bound for a class of KRR algorithms, which should be of independent interest. We start by introducing the learning model.

Let $\calh$ be a separable Hilbert space and $\mathcal{D} = \{(\x_i,y_i)\}_{i=1}^{n}$ be sampled from an unknown distribution $\rho$ over $\calh\times\R$. Let $\rho_{\calx}$ denote the marginal distribution over $\calh$ and $\rho(y|\x)$ denote the conditional distribution over $\y\in\R$.
We denote the inner product over $\calh$ by $\langle\cdot,\cdot\rangle_{\calh}$ and the induced norm by $\|\cdot\|_\calh$. Let $\calh_\rho$ be the Hilbert space consisting of linear forms over $\calh$, defined as $\calh_\rho = \{f:\calh \rightarrow \R \ | f(\x) = \langle \x,\omega\rangle_\calh  \ \text{for some} \  \omega \in \calh\}.$ We consider expected risk minimization: 

\begin{align}
    \inf_{f \in \calh_\rho}\mathcal{E}(f) &= \int_{\calh\times\R}{(f(\x)-y)^2d\rho(\x,y)}. \label{e:main}
\end{align}
The above formulation for nonparametric regression provides a general framework for learning over Hilbert spaces. In particular, it covers learning over reproducing kernel Hilbert spaces and has been studied extensively in literature \citep[e.g.,][]{cucker2007learning,caponnetto2007optimal,steinwart2008support,bauer2007regularization,lin2017optimal}. Note that the minimizer of expected risk in (\ref{e:main}) need not exist in $\calh_\rho$, leading to hard learning scenarios, known as unattainable learning settings. We next describe the standard  assumptions for this learning model.

\begin{assumption}[Bounded kernel assumption]\label{as:1}
    The support of $\rho_\calx$ is compact and there exists a constant $G >0$ such that $\langle \x,\x'\rangle_\calh \leq G^2  \ \forall \ \x,\x' \in \calh, \text{almost surely}$.  
\end{assumption}
\begin{assumption}[Moments assumption]\label{as:2}
    There exist constants $M,\sigma >0$ such that for any integer $p \geq 2$, we have $\int_{\R}{|y|^pd\rho(y|\x)} < \frac{1}{2}p!M^{p-2}\sigma^2$ \ \text{almost surely.}
   
\end{assumption}Consider the Hilbert space $L^2(\calh,\rho_\calx) := \{f: \calh \rightarrow \R \ | \int_{\calh}{f(\x)^2d\rho_\calx(\x) < \infty}\}$. The norm on $L^2(\calh,\rho_\calx)$ will be denoted by $\|\cdot\|_\rho$, and for any $f \in  L^2(\calh,\rho_\calx)$ we have $\|f\|_\rho = (f(\x)^2d\rho_\calx(\x))^{1/2}$. Let $\cals_\rho: \calh \rightarrow L^2(\calh,\rho_\calx)$ be the map $\cals_\rho\omega = \langle\cdot,\omega\rangle_\calh$ and $\cals_\rho^*: L^2(\calh,\rho_\calx) \rightarrow \calh$ be the adjoint map of $\cals_\rho$. The covariance operator $\cals_\rho^*\cals_\rho$ will be denoted as $\calc$, whereas the integral operator $\cals_\rho\cals_\rho^*$ will be denoted as $\call$. Let $\calz_n: \calh \rightarrow \R^n $ be $\calz_n\omega = \big(\langle\x_i,\omega\rangle\big)_{i=1}^{n}$. An important observation here is that the operator $\calz_n\calz_n^* : \R^n \rightarrow \R^n$  corresponds to  the $n\times n$ psd kernel matrix $\K=[\langle\x_i,\x_j\rangle_{\mathcal H}]
_{i,j}$ over the $n$ points $\x_i$ for $1\leq i \leq n$. The minimizer of $\mathcal{E}(f)$  over all functions in $L^2(\calh,\rho_\calx)$ is attained by the regression function $f_\rho$ \citep{cucker2007learning} defined as $f_\rho = \int_{\R}{yd\rho(y|\x)}$. 
The minimizer of (\ref{e:main}) is given by the orthogonal projection of $f_\rho$ onto the closure of $\calh_\rho$ in $L^2(\calh,\rho_\calx)$. The orthogonal projection of $f_\rho$ onto the closure of $\calh_\rho$ will be denoted by $f_\calh$.
\begin{assumption}[Smoothness assumption]\label{as:3}
    There exist positive constants $C,R,\zeta$ such that we have
    $\int_\calh{\big(f_\calh(\x)-f_\rho(\x)\big)^2\x\otimes\x d\rho_\calx(\x)} \preceq  C^2\cdot \calc$, and moreover,
    $f_\calh = \call^\zeta g$, with $\|g\|_\rho \leq R$.
\end{assumption}
Generally speaking, larger $\zeta$ corresponds to a more stringent assumption and implies that $f_\calh$ can be well approximated by a hypothesis in $\calh_\rho$. In particular, when $0<\zeta<\frac{1}{2}$, the function $f_\calh$ does not lie in $\calh_\rho$ and this corresponds to the harder unattainable setting. We focus on this setting, because when $\zeta\geq 1/2$, then existing Nystr\"om-based algorithms are essentially optimal \citep{rudi2018fast}.

\begin{assumption}[Capacity assumption]\label{as:4}
Let $\lambda>0$ and let $\mathcal{N}(\lambda)$ denote $\tr(\calc(\calc+\lambda\mathcal{I})^{-1})$, where $\mathcal{I}$ denotes the identity operator on $\calh$. Then there exists a positive constant $0\leq\gamma\leq 1$ and a constant $c_\gamma$ such that $\mathcal{N}(\lambda) \leq c_\gamma\lambda^{-\gamma}$.
\end{assumption}
Since $\calc$ is a trace class operator, the above assumption is always satisfied with $\gamma=1$ and $c_\gamma = G^2$. Moreover, if eigenvalues of $\calc$ decay polynomially i.e. $\sigma_i = O(i^{-1/\gamma})$, then the above assumption is satisfied with $c_\gamma = O(\gamma/(\gamma-1))$. The quantity $\mathcal{N}(\lambda)$ is known as the $\lambda$-effective dimension of the covariance operator $\calc$, and it is analogous to its matrix counterpart from Definition~\ref{d:ridge_lev}. 

\subsection{Expected risk of approximate KRR}
Consider the $\lambda$-regularized empirical risk defined as:
\begin{align}
    \tilde{\mathcal{E}}(\omega,\lambda) = \frac{1}{n}\sum_{i=1}^{n}{(\langle\omega,\x_i\rangle_{\calh}-y_i)^2} +\lambda\|\omega\|_\calh^2. \label{e:empirical_risk_reg}
\end{align}
A straightforward application of the representer theorem shows that the minimizer of $\tilde{\mathcal{E}}(\omega,\lambda)$ in $\calh$ lies in the $n$ dimensional space $\calh_n : = \big\{\sum_{i=1}^{n}{w_i\x_i} | \w\in \R^n\big\}$. Noting that $\calh_n$ is the range of operator $\calz_n^*$, it is easily shown that minimizing (\ref{e:empirical_risk_reg}) solves the linear system:  $(\K+n\lambda\I)\w=\y$. 

The Nystr\"om approximation of this problem can be interpreted as minimizing $\tilde{\mathcal{E}}(\omega,\lambda)$ over an $m$ dimensional subspace of $\calh_n$ defined by the $m\ll n$ Nystr\"om landmarks. For a random sub-sampling matrix $\S \in \R^{m\times n}$, the space $\calh_m$ is given as $\big\{\calz_n^*\S^\top\w | \w \in \R^m\big\}$. A simple exercise (see Lemma \ref{l:omega_rep}) shows that the minimizer of $\tilde{\mathcal{E}}(\omega,\lambda)$ is given by $\hat\omega = \hat\calp\calz_n^*(\calz_n\hat\calp\calz_n^*+n\lambda\I)^{-1}\y$, where $\hat\calp$ is the orthogonal projection operator onto the subspace $\calh_m$. Finding $(\calz_n\hat\calp\calz_n^*+n\lambda\I)^{-1}\y$ corresponds to solving $(\hat\K+n\lambda\I)\w=\y$, where $\hat\K = \calz_n\hat\calp\calz_n^* $ is the Nystr\"om approximation of~$\K$. The associated statistical risk with $\hat\omega$ is analyzed by upper bounding $\|\cals_\rho\hat\omega-f_\calh\|_\rho$, since for $\hat f = \cals_\rho\hat\omega$ one has $\mathcal{E}(\hat f) - \mathcal{E}(f_\calh) = \|\cals_\rho\hat\omega-f_\calh\|_\rho^2$ \citep{cucker2007learning, steinwart2008support}. 

While prior work \citep[e.g.,][]{lin2020optimal} has analyzed projection-based algorithms in this model, their analysis does not cover Block-Nystr\"om, because in this case $\hat\calp$ is not an orthogonal projection. To address this, we consider the following more general family of algorithms.
\begin{assumption}\label{as:5}
    Let $\hat\calp$ be an operator on space $\calh$ with properties: $0\preceq \hat\calp \preceq \mathcal{I}$ and $\calp\hat\calp=\hat\calp$, where $\calp$ denotes the orthogonal projection operator onto $\calh_n$ and $\mathcal{I}$ denotes the identity operator. 
    Then, for given $\lambda>0$, we consider the hypothesis $\hat\omega:=\hat\calp\calz_n^*(\calz_n\hat\calp\calz_n^*+n\lambda\I)^{-1}\y$.
\end{assumption}

This general learning model covers as special cases the models where $\hat\calp$ is an orthogonal projection, as in Nystr\"om KRR or sketching-based KRR methods. However, in the more general setup $\hat\calp$ need not be a projection. For example, in Block-Nystr\"om, we consider $\hat\calp = \frac{1}{q}\sum_{i=1}^{q}{\hat\calp_i}$ where $\hat\calp_i$ are independent random orthogonal projections. In particular, $\hat\calp_i = \calz_n^*\S_i^\top(\S_i\K\S_i^\top)^\dagger\S_i\calz_n$ where $\S_i \in \R^{m\times n}$ is a sub-sampling matrix, which yields Block-Nystr\"om KRR, $\hat\omega= \calp\calz_n^*(\hat\K_{[q]} +n\lambda\I)^{-1}\y$ where $\hat\K_{[q]} = \frac{1}{q}\sum_{i=1}^{q}{\hat\K_i}$ and $\hat\K_{i} = \calz_n\hat\calp_i\calz_n^*$ are iid Nystr\"om approximations.

\begin{theorem}[Expected risk of approximate KRR]
\label{t:main_stat_risk}
    Under assumptions \ref{as:1}, \ref{as:2}, \ref{as:3}, \ref{as:4}, and \ref{as:5}, let $\zeta<1/2$, $n \geq O(G^2\log(G^2/\delta))$ and $\frac{19G^2\log(n/\delta)}{n}\leq\lambda \leq \|\calc\|$. If $\hat\calp$ is constructed so that $\hat\K=\calz_n^*\hat\calp\calz_n$ is an $n\lambda$-regularized $\alpha$-approximation of $\K$,  as defined in \eqref{eq:approx}, then with high probability,
   \begin{align*}
    \big\|\cals_\rho \hat\omega- f_\calh\big\|_\rho \leq  \tilde O\Big(\alpha R\lambda^\zeta
    + \frac{1}{n\lambda^{1-\zeta}} + \frac{1}{\sqrt{n\lambda^{\gamma}}} \Big).
\end{align*}
\end{theorem}
The proof of Theorem \ref{t:main_stat_risk} (Appendix \ref{s:appendix_krr}) incorporates and extends techniques from \cite{rudi2015less}, \cite{lin2020optimal}, and \cite{lin2020convergences}. While previous works require a projection-based guarantee for the approximate KRR (such as Lemma \ref{l:musco-approx}), our proof only requires the looser regularized approximation guarantee~\eqref{eq:approx}. Note that the projection-based arguments upper bound the expected risk by approximating only the top part of the spectrum of $\K$, while our analysis also relies on how well the tail of the spectrum is approximated, thus leveraging the benefits of Block-Nystr\"om.

If we let $\alpha=O(1)$, then Theorem \ref{t:main_stat_risk} recovers the optimal learning rate for least squares regression over Hilbert spaces as proved in \cite{lin2020optimal} for $\zeta <1/2$. The optimal learning rate can be found by optimizing the right hand side over $\lambda$, leading to $\lambda^* \approx n^{-\frac{1}{\max\{1,2\zeta+\gamma\}}}$\footnote{Here, $\approx$ means upto constants depending on $G,\|\calc\|$ and $\log(n/\delta)$ factors}. This gives learning rate of $\tilde O(n^{-\rho^*})$, where $\rho^*=\zeta$ if $2\zeta+\gamma \leq 1$ and $\rho^*=\frac{\zeta}{2\zeta+\gamma}$ if $2\zeta+\gamma >1$. Considering $q=1$ and $\hat\K$ to be the Nystr\"om approximation of $\K$ sampled using $n\lambda^*$-ridge leverage scores of $\K$, we recover the guarantee for classical Nystr\"om KRR, as shown in \cite{lin2020convergences}. In the unattainable setting, the overall time complexity is dominated by sampling Nystr\"om landmarks according to $n\lambda^*$-ridge leverage scores of $\K$. Even using fast ridge leverage score sampling techniques from \cite{rudi2018fast}, the overall time complexity of running classical Nystr\"om KRR could be as high as $\tilde O\big(n^{\frac{1+2\gamma}{\max\{1,2\zeta+\gamma\}}}\big)$. For instance, if we consider the regime when $\gamma\approx 1$ and $\zeta$ is small, then the cost could be as large as $ O(n^3)$, completely erasing any computational gains.

To reduce this computational cost, suppose we allow larger $\alpha$, e.g., $\alpha=n^\theta$ for a small $\theta>0$. 
We show that, with the regularizer fixed to $n\lambda^*$, Block-Nystr\"om improves upon the time complexity of classical Nystr\"om KRR in this setting, while maintaining the excess risk factor of $\alpha$.

\begin{corollary}[Expected risk of Block-Nystr\"om KRR]\label{c:avg_nys_risk}
      Under Assumptions \ref{as:1}, \ref{as:2}, \ref{as:3} and \ref{as:4}, suppose that $\zeta<1/2$ and $n \geq \tilde O(G^2\log(G^2/\delta))$. Also, let $\lambda^* = \tilde O(n^{-\frac{1}{2\zeta+\gamma}})$ if $2\zeta+\gamma >1$ and $\lambda^* = \tilde O(n^{-1})$ if $2\zeta+\gamma \leq 1$. Consider Block-Nystr\"om KRR, $\hat\omega_{[q]}=\calp_{[q]}\calz_n^*(\hat\K_{[q]}+n\lambda^*\I)^{-1}\y$, constructed using $q=\tilde O(\alpha)$ blocks, each with $\tilde O(d_{\alpha^2\lambda^*}(\K))$ leverage score samples. Then,
      \begin{align*}
    \big\|\cals_\rho\hat\omega_{[q]}- f_\calh\big\|_\rho \leq \begin{cases} \alpha\cdot \tilde O(n^{-\frac{\zeta}{2\zeta+\gamma}})  &\text{if} \ \ 2\zeta+\gamma>1,\\
    \alpha\cdot \tilde O(n^{-\zeta}) &\text{if} \  \ 2\zeta+\gamma \leq 1.
\end{cases}
\end{align*}

\end{corollary}

\subsection{Cost comparison with classical Nystr\"om}
In Lemmas \ref{l:cost_single_Nystrom} and \ref{l:cost_avg_Nystrom} (Appendix \ref{s:appendix_krr}), we provide the time complexity analysis of Nystr\"om and Block-Nystr\"om KRR, for achieving near-optimal risk $\tilde O(n^{-\rho^*+\theta})$, where $\rho^*=\frac{\zeta}{\max\{1,2\zeta+\gamma\}}$ is the optimal rate and $\theta>0$ is a small constant (i.e., $\alpha=n^\theta$ in Corollary \ref{c:avg_nys_risk}). We show that when $\theta$ is sufficiently small and regularizer is $n\lambda^*$, the computational gains with Block-Nystr\"om KRR are at the order of $n^{2\theta\gamma}$ when compared with classical Nystr\"om KRR. In hard learning problems involving kernels with slow polynomial decay i.e., $\gamma \approx 1$, this yields a significant improvement of $\tilde O(n^{2\theta})$ over the existing projection-based KRR methods. Moreover, in Appendix~\ref{s:appendix_krr}, we further explore the computational gains of Block-Nystr\"om by considering all values of $\theta$ for which we can attain convergence. In particular, we observe that the optimized time complexity exhibits a phase transition between Block-Nystr\"om and classical Nystr\"om, which is attained by carefully reducing the number of blocks in Block-Nystr\"om while maintaining the~same~$\alpha$.

\section{Conclusions}

We propose Block-Nystr\"om, a low-rank approximation method for positive semidefinite matrices which is computationally faster than the Nystr\"om method for matrices with heavy-tailed spectrum. We show that Block-Nystr\"om speeds up preconditioning for convex optimization, and can improve the cost of near-optimal learning for kernel ridge regression  over Hilbert spaces.

\subsection*{Acknowledgments}
The authors would like to acknowledge support from NSF award CCF-2338655. Also, thanks to Aaron Sidford for valuable discussions and insights regarding recursive preconditioning.

\bibliography{arxiv}   
\appendix

\section{Auxiliary lemmas}
In this section, we mention the auxiliary results and lemmas that are used in our analysis.
\begin{lemma}[Adapted from \cite{js_talg}]\label{l:pcg}
 For any pd matrices $\A,\B \in \R^{n\times n}$ with $\A\preceq\B\preceq\kappa\A$
for $\kappa\geq1$, and $f$ that is a $\frac{1}{10\kappa}$-solver
for $\B$, i.e., given $\v\in\R^n$, returns $f(\v)$ such that $\|f(\b)-\B^{-1}\v\|_{\B}^2\leq\frac1{10\kappa}\|\B^{-1}\v\|_{\B}^2$, there is an $\epsilon$-solver for $\A$ that applies $f$ and $\A$ at most $\lceil 4 \sqrt{\kappa}\log(2/\epsilon)\rceil$ times each, and spends additional $O(n \sqrt{\kappa} \ln(1/\epsilon))$ time when run.
\end{lemma}

\begin{lemma}[\cite{askotch}]\label{rls-projection}
Let $\A\in\R^{n\times n}$ be psd and $\lambda'>0$. Let $\{p_i\}_{i=1}^n$ be $O(1)$-approximate $\lambda'$-ridge leverage score sampling probabilities. Let $\S \in \R^{O(d_{\lambda'}(\A)\log^3 n) \times n}$ be a random sub-sampling matrix, drawn from probability distribution $\{p_i\}_{i=1}^n$. Consider the random projection matrix $\P=\A^{1/2}\S^\top(\S\A\S^\top)^\dagger\S\A^{1/2}$. Then $\E[\P]$ satisfies 
\begin{align}
\E[\P] ~\succeq~ \frac12 \A \left(\A + \lambda'\I\right)^{-1}.
\end{align}
\end{lemma}

\begin{lemma}[Cordes inequality \cite{fujii1993norm}]\label{l:cordes}
Let $\mathcal{A}$ and $\mathcal{B}$ be two positive bounded linear operators on a
separable Hilbert space. Then for $0\leq s \leq 1$
\begin{align*}
    \|\mathcal{A}^s\mathcal{B}^s\| \leq \|\mathcal{A}\mathcal{B}\|^s.
\end{align*}
\end{lemma}

\begin{lemma}[Lemma 5 from \cite{rudi2015less}]\label{l:rudi_15_1}
     If assumptions \ref{as:1} and \ref{as:4} hold, then for any $\delta>0$, $n \geq O(G^2\log(G^2/\delta))$ and $\frac{19G^2\log(n/\delta)}{n}\leq\lambda \leq \|\calc\|$, we have,
     \begin{align*}
         \|\calc_{\lambda}^{1/2}\cdot\calc_{n\lambda}^{-1/2}\| < 2
     \end{align*}
     with probability $1-\delta$.
\end{lemma}
The assumption on $n$ can be removed, (see Lemma 5.3 in \cite{lin2020optimal}) but leads to a more complicated upper bound depending on $G^2$ in place of constant $2$. In this work, we are interested in a large $n$ regime and therefore assume $n \geq O(G^2\log(G^2/\delta))$ obtaining a constant upper bound.
\begin{lemma}[Proposition 1 from \cite{rudi2015less}]
If assumptions \ref{as:1} and \ref{as:4} hold, then for any $\delta>0$, $n =O(G^2\log(G^2/\delta))$ and $\frac{19G^2\log(n/\delta)}{n}\leq\lambda \leq \|\calc\|$, we have,
 \begin{align*}
     d_{n\lambda}(\K) \leq 3\mathcal{N}(\lambda)
 \end{align*}
with probability $1-\delta$.
\end{lemma}
The next two results are essentially taken from \cite{lin2020optimal} and upper bound the true bias and the sample variance terms in our analysis.

\begin{lemma}[Adapted from Lemma 5.2 in \cite{lin2020optimal}]\label{l:true_bias}
Let $0\leq a \leq \zeta$ and $\omega=\calc_\lambda^{-1}\cals_\rho^*f_\calh$. Under assumption \ref{as:3}
    \begin{align*}
        \big\|\call^{-a}(\cals_\rho\omega-f_\calh)\big\|_\rho \leq R\lambda^{\zeta-a}.
    \end{align*}
\end{lemma}

\begin{lemma}[Adapted from Lemma 5.6 in \cite{lin2020optimal}]\label{l:sample_variance} Let $\hat\y = \y/\sqrt{n}$ and $\omega=\calc_\lambda^{-1}\cals_\rho^*f_\calh$. Then under assumptions \ref{as:1}, \ref{as:2}, \ref{as:3}, \ref{as:4} for any $\delta>0$
    \begin{align*}
        \big\|\calc_{n\lambda}^{-1/2}[\cals_n^*\hat\y-\calc_n\omega]\big\|_{\calh} \leq C'\Big(\lambda^\zeta + \frac{1}{n\lambda^{\max(1/2,1-\zeta)}} + \frac{1}{\sqrt{n\lambda^{\gamma}}} \Big)\log(1/\delta)
    \end{align*}
    with probability $1-\delta$. The constant $C'$ depends on $R,G,M,C,c_\gamma,\zeta,$ and $\sigma^2$.
\end{lemma}

\section{Spectral approximation with Block-Nystr\"om}
\label{a:concentration}
We start this section by proving the spectral approximation guarantee (\ref{eq:approx}) provided by the Block-Nystr\"om approximation of $\A$.
\begin{theorem}[Spectral approximation with  Block-Nystr\"om (Restated Theorem \ref{t:Nys_spectral_approximation})]\label{t:Nys_spectral_approximation_2}
    Let $\lambda>0$ be fixed and $\lambda'>\lambda$. Let $\{p_i\}_{i=1}^n$ denote the $O(1)$-approximate $\lambda'$-ridge leverage scores of $\A$. Let $\hat\A_{[q]}=\frac{1}{q}\sum_{i=1}^{q}{\hat\A_i}$ where $\hat\A_i$'s are iid Nystr\"om approximations of $\A$ sampled independently from $\{p_i\}_{i=1}^n$, each with $\tilde O(d_{\lambda'}(\A)\log(n/\delta))$ landmarks and $\A_\lambda= \A+\lambda\I$.
    Then with probability $1-\delta$,
    \begin{align*}
\|\A_{\lambda}^{1/2}\big(\hat\A_{[q]} +\lambda\I\big)^{-1}\A_{\lambda}^{1/2}\| \leq 16\sqrt{\frac{\lambda'}{\lambda}}
    \end{align*}
    if $q > 64\sqrt{\frac{\lambda'}{\lambda}}\log(n/\delta)$.
\end{theorem}
\begin{proof}
    Define random matrices $\Z_i = \A_{\lambda}^{-1/2}(\hat\A_i +\lambda\I)\A_{\lambda}^{-1/2}$, where  $\hat\A_{i} = \A\S_i^\top(\S_i\A\S_i^\top)^\dagger\S_i\A $ and $\A_\lambda=\A+\lambda\I$. Here each $\S_i$ is a sub-sampling matrix corresponding to a set drawn independently from $\{1,2,\cdots,n\}$ with probability distribution $\{p_i\}_{i=1}^n$. As $\hat\A_i \preceq \A$ for all $i$, we have $\Z_i \preceq \I$ for all $i$. Therefore $\|\Z_i\| \leq 1:=R$. Let $\mu_{\min}$ denote the minimum eigenvalue of $\E\Z_i$. We have,
    \begin{align*}
        \mu_{\min} &= \lambda_{\min}(\A_{\lambda}^{-1/2}\E[\hat\A_i+\lambda\I]\A_{\lambda}^{-1/2})\\
&\geq\frac{1}{2}\lambda_{\min}\Big(\A_{\lambda}^{-1/2}\big(\A\A_{\lambda'}^{-1}\A+\lambda\I\big)\A_{\lambda}^{-1/2}\Big)\\
&=\frac{1}{2}\min_{j}\Big(\frac{1}{\lambda_j+\lambda}\cdot\big(\frac{\lambda_j^2}{\lambda_j+\lambda'}+\lambda\big)\Big)\\
&=\frac{1}{2}\min_{j}\Big(\frac{\lambda_j^2+\lambda\lambda_j+\lambda\lambda'}{\lambda_j^2+\lambda\lambda_j +\lambda'\lambda_j +\lambda\lambda'}\Big)
    \end{align*}
   The second inequality holds due to Lemma \ref{rls-projection} as $\E\hat\A_{i} \succeq \frac{1}{2}\A(\A+\lambda'\I)^{-1}\A$. Considering $\lambda_j$ as a variable $x$, a routine calculus exercise shows that the minimum is attained when $x = \sqrt{\lambda\lambda'}$. Therefore, we get,
    \begin{align*}
        \mu_{\min} \geq \frac{1}{2}\frac{\lambda\lambda' + \lambda\sqrt{\lambda\lambda'} + \lambda\lambda'}{\lambda\lambda' + \lambda\sqrt{\lambda\lambda'} + \lambda'\sqrt{\lambda\lambda'} + \lambda\lambda'} \geq \frac{1}{2}\frac{\lambda\lambda'}{4\lambda'\sqrt{\lambda\lambda'}} = \frac{1}{8}\sqrt{\frac{\lambda}{\lambda'}}.
    \end{align*}
    Applying matrix Chernoff concentration inequality we get
    \begin{align*}
        \Pr\Big(\lambda_{\min}\big(\frac{1}{q}\sum_{i=1}^{q}{\Z_i}\big) \leq (1-\epsilon)\mu_{\min}\Big) \leq n\cdot\exp\Big(-\frac{\epsilon^2q\mu_{\min}}{2R}\Big).
    \end{align*}
    Substitute $R=1$ and $\epsilon=\frac{1}{2}$, we get if $q>64\sqrt{\frac{\lambda'}{\lambda}}\log(n/\delta)$, then with
    probability $1-\delta$ $\lambda_{\min}\big(\frac{1}{q}\sum_{i=1}^{q}{\Z_i}\big) \geq \frac{\mu_{\min}}{2}.$
    This gives us
\vspace{-2mm}    \begin{align*}
        \frac{1}{16}\sqrt{\frac{\lambda}{\lambda'}}\A_{\lambda} \preceq \frac{1}{q}\sum_{i=1}^{q}{\hat\A_{i}+\lambda\I} \preceq \A_{\lambda},
    \end{align*}
and finishes the proof.
\end{proof}
The following result is needed in the proof of our recursive solver for inverting Block-Nystr\"om.
\begin{theorem}[Matrix concentration for Block-Nystr\"om]\label{t:exp_nys_con}
Let $\lambda'>0$ and for $1 \leq i \leq q$, let $\hat\A_i$ be Nystr\"om approximations of $\A$ sampled from $O(1)$-approximate $\lambda'$-ridge leverage scores of $\A$, each with $\tilde O(d_{\lambda'}(\A)\log(n/\delta))$ landmarks. Then for any $\lambda>0$ satisfying $\lambda \leq \lambda'$ and for any $0<\theta<1$, $q>\frac{200\sqrt{\lambda'/\lambda}\log(2n/\delta)}{\theta^2}$, with probability $1-\delta$ we get
\begin{align*}
   (1-\theta/2)(\E\hat\A+\lambda\I)\preceq   \frac{1}{q}\sum_{i=1}^{q}{\hat\A_i+\lambda\I} \preceq (1+\theta/2)(\E\hat\A +\lambda\I).
\end{align*}
\end{theorem}

\begin{proof}
First note that due to Lemma \ref{rls-projection},
$\big(\E\hat\A_{i} +\lambda\I\big)^{-1} \preceq 2\big(\A\A_{\lambda'}^{-1}\A+\lambda\I\big)^{-1}$. Define $\Z_i = (\E\hat\A+\lambda\I)^{-1/2}(\hat\A_i+\lambda\I)(\E\hat\A+\lambda\I)^{-1/2}$. We have $\E[\Z_i] =\I$ and therefore $\mu_{\max}(\E\Z_i) = \mu_{\min}(\E\Z_i)=1$, where $\mu_{\max}$ and $\mu_{\min}$ denotes the maximum and minimum eigenvalues of $\E\Z_i$. Now we upper bound $\|\Z_i\|$. Let $\A_\lambda=\A+\lambda\I$.
\begin{align*}
    \|\Z_i\|
    &\leq \|(\E\hat\A+\lambda\I)^{-1/2}(\A+\lambda\I)(\E\hat\A+\lambda\I)^{-1/2}\|\\
&\leq2\lambda_{\max}\Big(\A_\lambda^{1/2}(\E\hat\A+\lambda\I)^{-1}\A_\lambda^{1/2}\Big)\\
&\leq2\max_{i}\Big(\frac{\lambda_i+\lambda}{\frac{\lambda_i^2}{\lambda_i+\lambda'}+\lambda}\Big)\\
&=2\max_{i}\Big(\frac{1+\frac{\lambda}{\lambda_i}}{\frac{\lambda_i}{\lambda_i+\lambda'} + \frac{\lambda}{\lambda_i}}\Big) :=\xi_{i}.
\end{align*}
Note that if $\lambda_i\geq\lambda'$, then $\xi_i <3$. In case $\lambda_i < \lambda'$, we write $\lambda_i =c_i\lambda_n$ for some $c_i \geq 1$. We have
\begin{align*}
    \xi_i &<  \frac{1+\frac{\lambda}{\lambda_i}}{\frac{\lambda_i}{2\lambda'} + \frac{\lambda}{\lambda_i}}= \frac{1+\frac{\lambda}{c_i\lambda_n}}{\frac{c_i\lambda_n}{2\lambda'} + \frac{\lambda}{c_i\lambda_n}}
\end{align*}
and consider the following two cases:
\begin{enumerate}
    \item \emph{Case 1:} $\frac{c_i\lambda_n}{2\lambda'} \geq \frac{\lambda}{c_i\lambda_n}$ or $c_i \geq \frac{\sqrt{2\lambda\lambda'}}{\lambda_n}$. Then
    \begin{align*}
        \xi_i &< \frac{1+\frac{\lambda}{c_i\lambda_n}}{\frac{c_i\lambda_n}{2\lambda'}}< 1+ \frac{2\lambda'}{c_i\lambda_n}\leq 1+\sqrt{\frac{2\lambda'}{\lambda}}.
    \end{align*}
    \item \emph{Case 2:}  $\frac{c_i\lambda_n}{2\lambda'} < \frac{\lambda}{c_i\lambda_n}$ or $c_i < \frac{\sqrt{2\lambda\lambda'}}{\lambda_n}$. Then
    \begin{align*}
        \xi_i &<  \frac{1+\frac{\lambda}{c_i\lambda_n}}{\frac{\lambda}{c_i\lambda_n}}< 1+ \sqrt{\frac{2\lambda'}{\lambda}}.
    \end{align*}
\end{enumerate}
Therefore  $\|\Z_i\| < 2\max\Biggr\{3, 1+ \sqrt{\frac{2\lambda'}{\lambda}}\Biggr\} < 20\sqrt{\frac{\lambda'}{\lambda}}:=R$. Using matrix Chernoff inequality, for any $\theta >0$
\begin{align*}
    \Pr\left(\lambda_{\min}\left(\frac{1}{q}\sum_{i=1}^{q}{\Z_i}\right) \leq 1-\theta/2\right) \leq n.\exp\left(-\frac{q\theta^2}{8R}\right),
\end{align*}
and,
\begin{align*}
     \Pr\left(\lambda_{\max}\left(\frac{1}{q}\sum_{i=1}^{q}{\Z_i}\right) \geq 1+\theta/2\right) \leq n.\exp\left(-\frac{q\theta^2}{(8+2\theta)R}\right).
\end{align*}
For any $\theta<1$ and $q>200\sqrt{\lambda'/\lambda}\log(2n/\delta)/\theta^2$, with probability $1-\delta$
\begin{align*}
    -\theta/2\I\preceq \frac{1}{q}\sum_{i=1}^{q}{\Z_i} -\I \preceq \theta/2\I.
\end{align*}
and finally with probability $1-\delta$
\begin{align*}
   (1-\theta/2)(\E\hat\A+\lambda\I)\preceq   \frac{1}{q}\sum_{i=1}^{q}{\hat\A_{i}+\lambda\I} \preceq (1+\theta/2)(\E\hat\A +\lambda\I).
\end{align*}
\end{proof}
We use the following result from \cite{dmy24} as a black box in our recursive preconditioning scheme for solving the Block-Nystr\"om linear system approximately. The following lemma provides an algorithm to solve the resulting linear system at the final depth of our recursive scheme.
    \begin{lemma}[Based on Lemma 4.3 from \cite{dmy24}]\label{l:dmy_Nystrom_solve}
Given matrix $\A \succ 0$, let $\hat\A = \C\W^{\dagger}\C^\top$ be its Nystr\"om approximation where $\C=\A\S^\top \in \R^{n\times m}$ and $\W=\S\A\S^\top$. For $\psi>0$, denote $\M=\hat\A + \psi\I$ as the Nystr\"om preconditioner and assume that $\M\approx_{O(\kappa)} \A+\psi\I$. Given vector $\v \in \R^n$ and $\epsilon>0$, there exists an algorithm that provides $\hat\w$: an approximate solution to the linear system $\M\w=\v$ satisfying
\begin{align*}
    \|\hat\w-\M^{-1}\v\|_{\M} \leq \epsilon\|\M^{-1}\v\|_{\M}.
\end{align*}
in time $O((nm+m^3)\log(m/\delta) + (nm+m^2)\log(\kappa/\epsilon))$, where $\kappa$ is the condition number of $\A$.
    
\end{lemma}
The following lemma is a computational result we use in our proof of Theorem \ref{t:linear_system_solve_2}.
\begin{lemma}\label{l:c_bound}
For any $x,\phi,\theta >0$, $c>1$ and $p=c(1+\theta)^2$
\begin{align*}
(4(1+\theta)^2)^{\frac{\log(x)}{\log(c)}} &\leq x^{\phi/2}\\
\log(20p^2)^{\frac{\log(x)}{\log(c)}} &\leq (x)^{\phi/2}
    \end{align*}
\end{lemma}
if and only if $c > \max\big((4(1+\theta)^2)^{2/\phi},\log^{2/\phi}(20p^2),1\big)$.
\begin{proof}
    \begin{align*}
(4(1+\theta)^2)^{\frac{\log(x)}{\log(c)}} = x^{\frac{\log(4(1+\theta)^2)}{\log(c)}} \leq x^{\phi/2}
    \end{align*}
    if and only if $c \geq (4(1+\theta)^2)^{2/\phi}.$
    Furthermore for $p=c(1+\theta)^2$,
    \begin{align*}
\log(20p^2)^{\frac{\log(x)}{\log(c)}} \leq (x)^{\phi/2}
    \end{align*}
    if and only if $c \geq \log^{2/\phi}(20p^2)$.
\end{proof}
The following theorem is an extended version of Theorem \ref{t:linear_system_solve}. The statement of Theorem \ref{t:linear_system_solve} can be recovered by setting $\tilde\lambda=\lambda$. We provide a recursive preconditioning scheme for solving the Block-Nystr\"om linear system having near-linear dependence on the number of blocks $q$.
\begin{theorem}[Recursive solver for Block-Nystr\"om ]\label{t:linear_system_solve_2}
Given $\lambda,\tilde\lambda$ and $\lambda'$ satisfying $\lambda \leq \tilde\lambda \leq \lambda'$, for $1\leq i \leq q$ let $\hat\A_i$  are Nystr\"om approximations of $\A$ sampled from $O(1)$-approximate $\lambda'$-ridge leverage scores of $\A$, each with $m=O(d_{\lambda'}(\A)\log(n/\delta))$ landmarks and let $q = O\big(\sqrt{\frac{\lambda'}{\tilde\lambda}}\log(n/\delta)/\theta^2\big)$. Denote $\hat\A_{[q]} =  \frac{1}{q}\sum_{i=1}^{q}{\hat\A_{i}}$. There exists an event with probability $1-\delta$ for any $\delta>0$ such that for any $\v \in \R^n$, $\epsilon>0$ and $0<\phi<1$ we can find $\u \in \R^n$ such that
\begin{align*}
    \|\u-\big(\hat\A_{[q]} +\lambda\I\big)^{-1}\v\|_{\hat\A_{[q]} +\lambda\I}^2 \leq \epsilon\cdot \|\big(\hat\A_{[q]} +\lambda\I\big)^{-1}\v\|_{\hat\A_{[q]} +\lambda\I}^2
\end{align*}
in time $O\Big(\sqrt{\tilde\lambda/\lambda}\cdot nmq^{1+\phi}\cdot cs_1s_2k\log(1/\epsilon)\Big)$, where $k=\max\Big\{1,\Big\lceil\log\big(\sqrt{\lambda'/\tilde\lambda}\big)/\log(c)\Big\rceil\Big\}$, $s_1 = \max\big(1,\log\Big(\tilde\lambda/\lambda\Big)\big)$, $s_2 = \log(m\kappa/\delta)\log^{3/2}(n/\delta)\log(20c^2(1+\theta)^4)$, and $c$ is chosen so that \\ $c \geq \lceil \max\Big((4(1+\theta)^2)^{2/\phi}, \log^{2/\phi}(20c^2(1+\theta)^4\Big) \Big\rceil.$
    \end{theorem}

\begin{proof}
Let $\epsilon>0$, $\theta>0$ be fixed and some $c>1$  (to be determined later). Let $r_k = \lambda'/(c^{2k}\tilde\lambda)$ where $k\geq 0$ is an integer. Furthermore let $q_k = O(\sqrt{r_k}\log(n/\delta)/\theta^2)$. We have $\hat\A_{[q]} +\lambda\I =  \frac{1}{q_0}\sum_{i=1}^{q_0}{\hat\A_{i} + \lambda\I}$. Let $\M_k=\frac{1}{q_k}\sum_{i=1}^{q_k}{\hat\A_{i}+\frac{\lambda'}{r_k}\I}$. Note that
\begin{align*}
  (\lambda/\tilde\lambda)\M_0\preceq \hat\A_{[q]} +\lambda\I \preceq \M_0.
\end{align*}
Using Lemma \ref{l:pcg}, for any given $\v$ we can find $\u \in \R^n$ such that $\|\u-\big(\hat\A_{[q]} +\lambda\I\big)^{-1}\v\|_{\hat\A_{[q]} +\lambda\I}^2 \leq \epsilon\|\big(\hat\A_{[q]} +\lambda\I\big)^{-1}\v\|_{\hat\A_{[q]} +\lambda\I}^2$ in time $4\sqrt{\tilde\lambda/\lambda}\big(T(\M_0) + q_0nm\big)\log(2/\epsilon) + n\sqrt{\tilde\lambda/\lambda}\log(1/\epsilon)$, where $T(\M_0)$ is time complexity to find a $\u$ given $\v$ such that $\|\u-\M_0^{-1}\v\|_{\M_0}^{2} \leq \frac{\lambda}{10\tilde\lambda}\|\M_0^{-1}\v\|_{\M_0}^2$. The total time complexity is given as
\begin{align}
O\Big(\sqrt{\tilde\lambda/\lambda}\big(T(\M_0) + q_0nm)\log(1/\epsilon)\Big) \label{e:total_time_recursion}
\end{align}
Now note that according to Theorem \ref{t:exp_nys_con} and our choices of $q_0$ and $q_1$, with probability $1-\delta/n$
\begin{align*}
    \M_0 \preceq (1+\theta)\Big(\E\hat\A + \frac{\lambda'}{r_0}\I\Big) \preceq (1+\theta)\Big(\E\hat\A + \frac{\lambda'}{r_1}\I\Big)\preceq (1+\theta)^2\M_1.
\end{align*}
Furthermore,
\begin{align*}
    \M_0 \succeq \frac{1}{1+\theta}\Big(\E\hat\A + \frac{\lambda'}{r_0}\I\Big)\succeq \frac{1}{c^2(1+\theta)}\Big(\E\hat\A + \frac{\lambda'}{r_1}\I\Big) \succeq \frac{1}{c^2(1+\theta)^2}\M_1.
\end{align*}
Again using Lemma \ref{l:pcg} we get
\begin{align}
    T(\M_0) &\leq 4c(1+\theta)^2\big(T(\M_1) + q_0nm\big)\log(20\tilde\lambda/\lambda) + nc(1+\theta)^2\log(10\tilde\lambda/\lambda) \nonumber\\
    &=O\Big((T(\M_1) + q_0nm)cs_1\Big).\label{e:tm0}
\end{align}
where $s_1 = \log(20\tilde\lambda/\lambda)$. Here $T(\M_1)$ is time complexity to find a $\u$ given $\v$ such that $\|\u-\M_1^{-1}\v\|^2_{\M_1} \leq \frac{1}{10c^2(1+\theta)^4}\|\M_1^{-1}\v\|^2_{\M_1}$. Similar as above we can show that
\begin{align*}
   \frac{1}{c^2(1+\theta)^2}\M_2 \preceq \M_1 \preceq (1+\theta)^2\M_2
\end{align*}
upper bounding $T(\M_1)$ as
\begin{align*}
    T(\M_1) \leq 4c(1+\theta)^2\big(T(\M_2) + q_1nm\big)\log(20c^2(1+\theta)^4) + nc(1+\theta)^2\log(10c^2(1+\theta)^4).
\end{align*}
Now recursively continuing in this manner by conditioning over high probability events $ \frac{1}{1+\theta}\E[\M_k]\preceq \M_k \preceq (1+\theta)\E[\M_k]$ according to Theorem \ref{t:exp_nys_con}, we get
\begin{align*}
    T(\M_k) \leq 4c(1+\theta)^2\big(T(\M_{k+1}) + q_knm\big)\log(20c^2(1+\theta)^4) + nc(1+\theta)^2\log(10c^2(1+\theta)^4).
\end{align*}
In fact, we can show that for $p=c(1+\theta)^2$
\begin{align*}
    T(\M_1) &\leq 4^kp^k\log^k(20p^2)T(\M_{k+1}) + nm\sum_{i=1}^{k}{4^{i}q_{i}p^i\log^{i}(20p^2)} + \frac{n}{4}\sum_{i=1}^{k}{4^{i}p^{i}\log^i(20p^2)} \nonumber\\
     &\leq 4^kp^k\log^k(20p^2)T(\M_{k+1}) + 2nmq_0\sum_{i=1}^{k}{\Big(\frac{4p\log(20p^2)}{c}\Big)^{i}}\\
     &\leq  4^kp^k\log^k(20p^2)T(\M_{k+1}) + 2knmq_0\cdot\big(4(1+\theta)^2\big)^k\cdot\big(\log(20p^2)\big)^k.
\end{align*}
In second to last inequality we used $q_i \geq 1$ for all $i$ and the formula $q_i = q_0/c^{i}$. Choose $k$ such that $r_{k+1}<2$ i.e., $k= \Big\lceil\frac{\log\sqrt{\lambda'/\tilde\lambda}}{\log(c)}\Big\rceil$ is sufficient. Now, let $c \geq \max\Big\{(4(1+\theta)^2)^{2/\phi}, \log^{2/\phi}(20p^2)\Big\}.$ Using Lemma \ref{l:c_bound} we get
\begin{align}
    T(\M_1) &\leq  4^kp^k\log^k(20p^2)T(\M_{k+1}) + 2knmq_0\cdot\Big(\sqrt{\frac{\lambda'}{\tilde\lambda}}\Big)^\phi \nonumber \\
    &= \frac{q_0}{q_k}\cdot\big(4(1+\theta)^2\big)^k\cdot\big(\log(20p^2)\big)^k\cdot T(\M_{k+1})   + 2knmq_0\cdot\Big(\sqrt{\frac{\lambda'}{\tilde\lambda}}\Big)^\phi \nonumber\\
    &\leq \Big(T(\M_{k+1}) + 2knm\Big)\cdot q_0^{1+\phi} \label{e:recursion}
\end{align}
where in the last inequality we used $\sqrt{\frac{\lambda'}{\tilde\lambda}} < q_0$. Now we precondition $\M_{k+1}$ with $\hat\A_1 + \frac{\lambda'}{r_{k+1}}\I$. As we have
\begin{align*}
    \frac{1}{q_{k+1}}\Big(\hat\A_1 + \frac{\lambda'}{r_{k+1}}\I\Big) \preceq \M_{k+1} \preceq \A+\frac{\lambda'}{r_{k+1}}\I \preceq \A+\lambda'\I \preceq 2(\hat\A_1+\lambda'\I) \preceq 4\Big(\hat\A_1+\frac{\lambda'}{r_{k+1}}\I\Big).
\end{align*}
where the second last inequality holds with probability $1-\delta/n$, due to Lemma \ref{l:musco-approx}. Using Lemma \ref{l:pcg} we get
\begin{align*}
    T(\M_{k+1}) \leq 8\sqrt{q_{k+1}}\big(T(\hat\A_1) + q_{k+1}nm\big)\log(20p^2) + 2n\sqrt{q_{k+1}}\log(10p^2)
\end{align*}
where $T(\hat\A_1)$ is the time complexity to find a $\u$ given $\v$ such that $\Big\|\u-\Big(\hat\A_1+\frac{\lambda'}{r_{k+1}}\I\Big)^{-1}\v\Big\|^2_{\hat\A_1+\lambda'/r_{k+1}\I} \leq \frac{1}{40q_{k+1}}\Big\|\Big(\hat\A_1+\frac{\lambda'}{r_{k+1}}\I\Big)^{-1}\v\Big\|^2_{\hat\A_1+\lambda'/r_{k+1}\I}.$ It only remains to upper bound $T(\hat\A_1)$. We apply Lemma \ref{l:dmy_Nystrom_solve} to upper bound $T(\hat\A_1)$. Let $\psi= \lambda'/r_k$ and apply Lemma \ref{l:dmy_Nystrom_solve} with $\M=\hat\A_1 + \psi\I$, we get that $\M^{-1}\x$ can be solved approximately in time $O\Big(nm\log(m\kappa/\delta)\Big)$, thus upper bounding $T(\hat\A_1)$. Substituting this in expression for $T(\M_{k+1})$ we get
\begin{align*}
    T(\M_{k+1}) &= O\Big((q_{k+1}nm)\log(m\kappa/\delta)\sqrt{q_{k+1}}\log(20p^2)\Big)\\
    &=O\Big(nm\log(m\kappa/\delta)\log^{3/2}(n/\delta)\log(20p^2)\Big)
\end{align*}
Denoting $s_2 = \log(m\kappa/\delta)\log^{3/2}(n/\delta)\log(20p^2)$ and substituting the upper bound for $T(\M_{k+1})$ in (\ref{e:recursion}) we get
\begin{align*}
    T(\M_1) =O\Big(\big(nms_2 + 2knm\big)\cdot q_0^{1+\phi}\Big).
\end{align*}
Back substituting the above in (\ref{e:tm0}) we get
\begin{align*}
    T(\M_0) = O\Big(\big(nms_2 + 2knm\big)\cdot q_0^{1+\phi}\cdot cs_1 + nmq_0\cdot cs_1\Big)
\end{align*}
Finally substituting the above in expression (\ref{e:total_time_recursion}), the overall time complexity is given as
\begin{align*}
&O\Big(\sqrt{\tilde\lambda/\lambda}\big(\big(nms_2 + 2knm\big)\cdot q_0^{1+\phi}\cdot cs_1 + nmq_0\cdot cs_1 + q_0nm\big)\log(1/\epsilon)\Big)\\
=&O\Big(\sqrt{\tilde\lambda/\lambda}\cdot nmq_0^{1+\phi}\cdot cs_1s_2k\log(1/\epsilon)).
\end{align*}
Taking the union bound over all the considered high probability events finishes the proof.
\end{proof}

\section{Statistical risk of KRR in random design setting}\label{s:appendix_krr}
Consider the Hilbert space $L^2(\calh,\rho_\calx)$ consisting of square integrable functions from $\calh$ to $\R$, i.e., $L^2(\calh,\rho_\calx) = \{f: \calh \rightarrow \R \ | \int_{\calh}{f(\x)^2d\rho_\calx(\x) < \infty}\}$. The norm on $L^2(\calh,\rho_\calx)$ will be denoted by $\|\cdot\|_\rho$, and for any $f \in  L^2(\calh,\rho_\calx)$ we have $\|f\|_\rho = \Big(f(\x)^2d\rho_\calx(\x)\Big)^{1/2}$. Let $\cals_\rho: \calh \rightarrow L^2(\calh,\rho_\calx)$ be the map $\cals_\rho\omega = \langle\cdot,\omega\rangle_\calh$ and $\cals_\rho^*: L^2(\calh,\rho_\calx) \rightarrow \calh$ be the adjoint map of $\cals_\rho$. The covariance operator $\cals_\rho^*\cals_\rho : \calh \rightarrow \calh$ will be denoted as $\calc$, whereas the integral operator $\cals_\rho\cals_\rho^* : L^2(\calh,\rho_\calx) \rightarrow L^2(\calh,\rho_\calx)$ will be denoted as $\call$. Under assumption \ref{as:1}, $\calc$ and $\call$ are positive, self-adjoint, and bounded linear operators. Let $\calz_n: \calh \rightarrow \R^n $ as $\calz_n\omega = \big(\langle\x_i,\omega\rangle\big)_{i=1}^{n}$ and $\cals_n = \calz_n/\sqrt{n}$. The operator $\cals_n^*\cals_n: \calh\rightarrow\calh$ is known as the empirical covariance operator and will be denoted by $\calc_n$. Note that $\calz_n^*\calz_n = n\calc_n$. An important observation here is that the operator $\calz_n\calz_n^* : \R^n \rightarrow \R^n$ is the linear map $\alpha \rightarrow \K\alpha$ where $\K \in \R^{n\times n}$ is the kernel matrix over the $n$ points $\x_i$ for $1\leq i \leq n$. For $\lambda>0$, the regularized operators $\calc+\lambda\mathcal{I}$, $\calc_n+\lambda\mathcal{I}$ will be denoted by $\calc_\lambda$ and $\calc_{n\lambda}$ respectively.
The following provides further intuition on the various operators considered in the analysis of statistical risk of KRR:
\begin{itemize}
    \item For any $g \in L^2(\calh,\rho_\calx)$, $\cals_\rho^* g = \int_\calh{\x g(\x)d\rho_\calx(\x)}$, $\call g = \int_\calh{\langle\cdot,\x\rangle_\calh g(\x)d\rho_\calx(\x)}$.
    \item For any $\omega \in \calh$, $\cals_\rho^*\cals_\rho\omega = \int_{\calh}{\langle\cdot,\omega\rangle_\calh\x d\rho_\calx(\x)}$, $\cals_n^*\cals_n\omega = \frac{1}{n}\sum_{i=1}^{n}{\langle\x_i,\omega\rangle_\calh\x_i}.$
    \item For any $\alpha \in \R^n$, $\calz_n^*\alpha = \sum_{i=1}^{n}{\alpha_i\x_i}$ and $\calz_n\calz_n^*\alpha = \K\alpha$.
\end{itemize}
The following lemma provide two different ways to represent the hypothesis $\hat\omega$ learned by the classical Nystr\"om KRR.
\begin{lemma}[Representations of learned hypothesis]\label{l:omega_rep}
Let $\hat\K$ be an $m$-rank Nystr\"om approximation of $\K$. Let $\hat\calp$ denote the projection operator onto the corresponding subspace $\calh_m$. Then the hypothesis learned by classical Nystr\"om KRR is given as
    \begin{align*}
        \hat\omega &= \hat\calp\calz_n^*\big(\hat\K+n\lambda\I\big)^{-1}\y\\
        &=\hat\calp\calz_n^*(\calz_n\hat\calp\calz_n^* + n\lambda\I)^{-1}\y\\
&=\hat\calp(\hat\calp\calz_n^*\calz_n\hat\calp +n\lambda\mathcal{I})^{-1}\hat\calp\calz_n^*\y.
    \end{align*}
\end{lemma}
\begin{proof}
Recall that
\begin{align*}
    \tilde{\mathcal{E}}(\omega,\lambda) = \frac{1}{n}\sum_{i=1}^{n}{(\langle\omega,\x_i\rangle_{\calh}-y_i)^2} +\lambda\|\omega\|_\calh^2.
\end{align*}
Let $\hat\omega$ be the minimizer of $  \tilde{\mathcal{E}}(\omega,\lambda)$. First note $\calz_n\hat\omega = \big(\langle\hat\omega,\x_i\rangle_{\calh}\big)_{i=1}^{n}$. Substituting $\hat\omega=\calz_n^*\S^\top\hat\w$ where $\S\in\R^{m\times n}$ is a random sketching matrix, we get $\calz_n\calz_n^*\S^\top\hat\w = \big(\langle\hat\omega,\x_i\rangle_{\calh}\big)_{i=1}^{n}$. Therefore $\sum_{i=1}^{n}{(\langle\hat\omega,\x_i\rangle_{\calh}-y_i)^2} = \|\K\S^\top\hat\w-\y\|^2$, by noting $\calz_n\calz_n^*=\K$. In particular $\hat\w$ is the solution of the system $(\S\K^2\S^\top +n\lambda\S\K\S^\top)\hat\w=\S\K\y$, implying,
    \begin{align*}
      \hat\w =  (\S\K^2\S^\top +n\lambda\S\K\S^\top)^\dagger\S\K\y.
    \end{align*}
    The learned $\hat\omega$ is given as
    \begin{align*}
        \hat\omega &= \calz_n^*\S^\top(\S\K^2\S^\top +n\lambda\S\K\S^\top)^\dagger\S\K\y\\
        &=\calz_n^*\S^\top(\S\calz_n(\calz_n^*\calz_n +n\lambda\mathcal{I})\calz_n^*\S^\top)^\dagger\S\calz_n\calz_n^*\y\\
&=\calz_n^*\S^\top(\S\calz_n(\S\calz_n^\dagger)\S\calz_n(\calz_n^*\calz_n +n\lambda\mathcal{I})\calz_n^*\S^\top(\calz_n^*\S^\top)^\dagger\calz_n^*\S^\top)^\dagger\S\calz_n\calz_n^*\y\\
&=\calz_n^*\S^\top(\S\calz_n(\hat\calp\calz_n^*\calz_n\hat\calp+n\lambda\mathcal{I})\calz_n^*\S^\top)^\dagger\S\calz_n\calz_n^*\y\\
&=\hat\calp(\hat\calp\calz_n^*\calz_n\hat\calp+n\lambda\mathcal{I})^{-1}\hat\calp\calz_n^*\y.
\end{align*}
The second last equality holds because $\calz_n^*\S^\top(\calz_n^*\S^\top)^\dagger$ and $(\S\calz_n)^\dagger\S\calz_n$ are orthogonal projections on $\calh_m$. Furthermore using the push-through identity we have $(\hat\calp\calz_n^*\calz_n\hat\calp+n\lambda\mathcal{I})^{-1}\hat\calp\calz_n^* = \hat\calp\calz_n^*(\calz_n^*\hat\calp\calz_n+n\lambda\I)^{-1}\y$ leading to
\begin{align*}
\hat\omega&=\hat\calp\calz_n^*(\calz_n\hat\calp\calz_n^*+n\lambda\I)^{-1}\y\\
&=\hat\calp\calz_n^*\big(\hat\K +n\lambda\I\big)^{-1}\y.
\end{align*}
The last equality holds because $\hat\K=\calz_n\hat\calp\calz_n^*$.
\end{proof}
A similar representation like Lemma \ref{l:omega_rep} holds for Block-Nystr\"om KRR as shown before the proof of Theorem \ref{t:main_stat_risk_2}. In the next result, we first show that for any given $\x$ in $\calh$, the output $y$ can be predicted in time linear in number of blocks, provided the Block Nystr\"om hypothesis has been precomputed.

\begin{lemma}[Prediction with Block-Nystr\"om]\label{l:prediction}
      For some $q\geq 1$, consider $\calp_{[q]} = \frac{1}{q}\sum_{i=1}^{q}{\hat\calp_i}$ where $\hat\calp_i$ denotes an orthogonal projection onto some $m$ dimensional subspace of $\calh_n$. In particular, $\hat\calp_i = \calz_n^*\S_i^\top(\S_i\K\S_i^\top)^\dagger\S_i\calz_n$ where $\S_i \in \R^{m\times n}$ is a sub-sampling matrix. Let $\hat\omega_{[q]} = \calp_{[q]}\calz_n^*(\hat\K_{[q]} +n\lambda\I)^{-1}\y$ where $\hat\K_{[q]} = \frac{1}{q}\sum_{i=1}^{q}{\hat\K_i+n\lambda\I}$ and $\hat\K_{i} = \calz_n\hat\calp_i\calz_n^*$ are iid Nystr\"om approximations of $\K$. Then for any $\x \in \calh$,  $\langle\hat\omega_{[q]},\x\rangle_\calh$ can be computed in time $O(qm)$ after a preprocessing cost of $\tilde O\Big(qm^3 + nmq^{1+o(1)}\Big)$, assuming $(\hat\K_{[q]} +n\lambda\I)^{-1}\y$ can be computed in time $\tilde O(nmq^{1+o(1)}).$
\end{lemma}
\begin{proof} Given any $\x \in \calh$, we can predict the response $y$ for $\x$ by computing $\langle\hat\omega_{[q]},\x\rangle_\calh$ as follows.
\begin{align*}
    \langle\omega_{[q]},\x\rangle_\calh &= \frac{1}{q}\sum_{i=1}^{q}{\langle\hat\calp_i\calz_n^*(\hat\K_{[q]}+n\lambda\I)^{-1}\y,\x\rangle_\calh}\\
    &=\frac{1}{q}\sum_{i=1}^{q}{\langle\calz_n^*\S_i^\top(\S_i\K\S_i^\top)^\dagger\S_i\calz_n\calz_n^*(\hat\K_{[q]}+n\lambda\I)^{-1}\y,\x\rangle_\calh}\\
    &=\frac{1}{q}\sum_{i=1}^{q}{\langle(\S_i\K\S_i^\top)^\dagger\S_i\calz_n\calz_n^*(\hat\K_{[q]}+n\lambda\I)^{-1}\y,\S_i\calz_n\x\rangle}.
\end{align*}
For any fixed $i$, $(\S_i\K\S_i^\top)^\dagger$ can be computed in time $O(m^3)$, $\S_i\calz_n\calz_n^*$ in time $O(nm)$\footnote{We are assuming that $\langle\x_i,\x_j\rangle_\calh$ takes times $O(1)$ for any $1\leq i,j \leq n$.}, and $\S_i\calz_n\calz_n^*(\hat\K_{[q]}+n\lambda\I)^{-1}\y$ in time $O(nm)$ given $\S_i\calz_n\calz_n^*$ and $(\hat\K_{[q]}+n\lambda\I)^{-1}\y$ have been precomputed. For any $\x \in \calh$, $\S_i\calz_n\x$ can be computed in time $O(m)$. The overall preprocessing cost to find $\langle\hat\omega_{[q]},\x\rangle_\calh$ is $\tilde O\Big(qm^3 + nmq^{1+o(1)}\Big)$.
\end{proof}
In Theorem \ref{t:linear_system_solve_2}, we provide a recursive algorithm to approximate $(\hat\K_{[q]}+n\lambda\I)^{-1}\y$ in time $ \tilde O(nmq^{1+o(1)})$ with high precision, for an appropriately chosen $q$ depending on $\lambda$ and sampling distribution for Nystr\"om landmarks. We now provide the proof of statistical risk of Block-Nystr\"om KRR. For $q\geq 1$ and $1\leq i \leq q$, consider $\hat\calp_{i}$ to denote orthogonal projections on $m$ dimensional subspaces of $\calh_n$ sampled in an iid manner and $\hat\K_i = \calz_n\hat\calp_i\calz_n^*$. Let $\hat\calp_{[q]}$ denote the average of $q$ projections i.e., $\hat\calp_{[q]} = \frac{1}{q}\sum_{i=1}^{q}{\hat\calp_i}$ and $\hat\K_{[q]}= \frac{1}{q}\sum_{i=1}^{q}{\hat\K_i}=\calz_n\hat\calp_{[q]}\calz_n^*$. In  Block-Nystr\"om, we consider the following hypothesis in $\calh$
\begin{align*}
    \omega_{[q]} &= \hat\calp_{[q]}\calz_n^*(\hat\K_{[q]} +n\lambda\I)^{-1}\y\\
    &=\hat\calp_{[q]}\calz_n^*(\calz_n\hat\calp_{[q]}\calz_n^* +n\lambda\I)^{-1}\y\\
&=\hat\calp_{[q]}^{1/2}\big(\hat\calp_{[q]}^{1/2}\calz_n^*\calz_n\hat\calp_{[q]}^{1/2}+n\lambda\mathcal{I}\big)^{-1}\hat\calp_{[q]}^{1/2}\calz_n^*\y.
\end{align*}
Recall $f_\rho$ denotes the regression function and $f_\calh$ denotes the projection of $f_\rho$ onto $\bar\calh_\rho$. We define $\omega = \calc_\lambda^{-1}\cals_\rho^*f_\calh$. 

The following theorem provides a more general version of Theorem \ref{t:main_stat_risk}, bounding the generalization error in terms of a range of norms parameterized by $a$, following prior works \citep{lin2020convergences,lin2020optimal}. The bounds in terms of the expected risk stated in Theorem \ref{t:main_stat_risk} can be recovered by taking $a=0$.

\begin{theorem}[Expected risk of approximate KRR, extended Theorem \ref{t:main_stat_risk}]\label{t:main_stat_risk_2}
Under assumptions  \ref{as:1}, \ref{as:2}, \ref{as:3}, \ref{as:4}, \ref{as:5},  $0\leq a \leq \zeta$, $\zeta < \frac{1}{2}$, $n \geq O(G^2\log(G^2/\delta))$ and $\frac{19G^2\log(n/\delta)}{n}\leq\lambda \leq \|\calc\|$. 
   \begin{align*}
    \big\|\call^{-a}(\cals_\rho \omega_{[q]}- f_\calh)\big\|_\rho \leq   \lambda^{-a}\cdot \tilde O\Big(R\lambda^\zeta\big\|\K_{n\lambda}^{1/2}(\hat\K_{[q]}+n\lambda\I)^{-1}\K_{n\lambda}^{1/2}\big\| + \frac{1}{n\lambda^{1-\zeta}} + \frac{1}{\sqrt{n\lambda^{\gamma}}} \Big)
\end{align*}
where $\K_{n\lambda} = \K+n\lambda\I$.
\end{theorem}
\begin{proof}
For simplicity of notation we will denote  $(\hat\K_{[q]}+n\lambda\I) = \bar\K_{n\lambda}$, $\hat\calp_{[q]} = \bar\calp$. We start with the following decomposition:
    \begin{align}
        \big\|\call^{-a}\big(\cals_\rho\omega_{[q]}-f_\calh\big)\big\|_\rho &\leq \big\|\call^{-a}\cals_\rho\big(\omega_{[q]}-\omega\big)\big\|_\rho + \underbrace{\big\|\call^{-a}\big(\cals_\rho\omega-f_\calh)\big\|_\rho}_{\text{True Bias}} \nonumber\\
        &\leq \big\|\call^{-a}\cals_\rho\big(\omega_{[q]}-\omega\big)\big\|_\rho + R\lambda^{\zeta-a} \label{e:risk_1}
    \end{align}
    where in last inequality we used Lemma \ref{l:true_bias} to upper bound the true bias term. Consider the first term now:
    \begin{align}
        \big\|\call^{-a}\cals_\rho\big(\omega_{[q]}-\omega\big)\big\|_\rho &= \big\|\call^{-a}\cals_\rho\calc^{a-1/2}\calc^{1/2-a}\big(\omega_{[q]}-\omega\big)\big\|_\rho \nonumber\\
        &\leq \big\|\call^{-a}\cals_\rho\calc^{a-1/2}\big\|\cdot\big\|\calc^{1/2-a}\big(\omega_{[q]}-\omega\big)\big\|_\calh \nonumber\\
        &\leq \big\|\calc^{1/2-a}\big(\omega_{[q]}-\omega\big)\big\|_\calh. \label{e:risk_2}
    \end{align}
The last inequality is obtained after upperbounding $\big\|\call^{-a}\cals_\rho\calc^{a-1/2}\big\|$ by $1$. Now we substitute for $\omega_{[q]}$. Continuing we have,
\begin{align}
     \big\|\calc^{1/2-a}\big(\omega_{[q]}-\omega\big)\big\|_\calh &= \big\|\calc^{1/2-a}\big[\bar\calp\calz_n^*\bar\K_{n\lambda}^{-1}\y - \omega\big]\big\|_\calh \nonumber\\
     &=\big\|\calc^{1/2-a}\big[\bar\calp\calz_n^*\bar\K_{n\lambda}^{-1}(\y-\calz_n\omega +\calz_n\omega)-\omega\big]\big\|_\calh \nonumber\\
     &\leq \underbrace{\big\|\calc^{1/2-a}\big[\bar\calp\calz_n^*\bar\K_{n\lambda}^{-1}(\y-\calz_n\omega)\big]\big\|_\calh}_{\mathcal{T}_v: \text{ Variance}} + \underbrace{\big\|\calc^{1/2-a}\big[\bar\calp\calz_n^*\bar\K_{n\lambda}^{-1}\calz_n\omega - \omega\big]\big\|_\calh}_{\mathcal{T}_b: \text{ Bias}}. \label{e:risk_3}
\end{align}
We upper bound the bias and variance terms separately, considering the variance term first:
\begin{align*}
    \mathcal{T}_v &=\big\|\calc^{1/2-a}\big[\bar\calp\calz_n^*\bar\K_{n\lambda}^{-1}(\y-\calz_n\omega)\big]\big\|_\calh\\
 &=\big\|\calc^{1/2-a}\calc_\lambda^{a-1/2}\calc_\lambda^{1/2-a}\calc_{n\lambda}^{a-1/2}\calc_{n\lambda}^{1/2-a}\big[\bar\calp\calz_n^*\bar\K_{n\lambda}^{-1}(\y-\calz_n\omega)\big]\big\|_\calh\\
 &\leq \big\|\calc^{1/2-a}\calc_\lambda^{a-1/2}\big\|\cdot\big\|\calc_\lambda^{1/2-a}\calc_{n\lambda}^{a-1/2}\big\|\cdot\big\|\calc_{n\lambda}^{-a}\big\|\cdot\big\|\calc_{n\lambda}^{1/2}\big[\bar\calp\calz_n^*\bar\K_{n\lambda}^{-1}(\y-\calz_n\omega)\big]\big\|_\calh.
\end{align*}
We have $\big\|\calc^{1/2-a}\calc_\lambda^{a-1/2}\big\| \leq 1$, and $\big\|\calc_\lambda^{1/2-a}\calc_{n\lambda}^{a-1/2}\big\| \leq \big\|\calc_{\lambda}^{1/2}\calc_{n\lambda}^{-1/2}\big\|^{1-2a} \leq 2$ by using Lemma \ref{l:cordes} and Lemma \ref{l:rudi_15_1}. Therefore we get,
\begin{align*}
    \mathcal{T}_v &\leq 2\lambda^{-a}\big\|\calc_{n\lambda}^{1/2}\big[\bar\calp\calz_n^*\bar\K_{n\lambda}^{-1}(\y-\calz_n\omega)\big]\big\|_\calh\\
    &=2\lambda^{-a}\big\|\calc_{n\lambda}^{1/2}\big[\bar\calp^{1/2}\bar\calp^{1/2}\calz_n^*(\calz_n\bar\calp^{1/2}\bar\calp^{1/2}\calz_n^*+n\lambda\I)^{-1}(\y-\calz_n\omega)\big]\big\|_\calh.
\end{align*}
In the last equality we substituted $\bar\K_{n\lambda}^{-1} = (\calz_n\bar\calp\calz_n^*+n\lambda\I)^{-1}$. Using push-through identity we have $\bar\calp^{1/2}\calz_n^*(\calz_n\bar\calp^{1/2}\bar\calp^{1/2}\calz_n^*+n\lambda\I)^{-1} = (\bar\calp^{1/2}\calz_n^*\calz_n\bar\calp^{1/2} +n\lambda\mathcal{I})^{-1}\bar\calp^{1/2}\calz_n^*$. Substituting this in the last equality we have,
\begin{align*}
     \mathcal{T}_v &\leq 2\lambda^{-a}\big\|\calc_{n\lambda}^{1/2}\bar\calp^{1/2}(\bar\calp^{1/2}\calz_n^*\calz_n\bar\calp^{1/2} +n\lambda\mathcal{I})^{-1}\calp^{1/2}\big[\calz_n^*\y - \calz_n^*\calz_n\omega\big]\big\|_\calh\\
     &=2\lambda^{-a}\big\|\calc_{n\lambda}^{1/2}\bar\calp^{1/2}(\bar\calp^{1/2}\calc_n\bar\calp^{1/2} +\lambda\mathcal{I})^{-1}\calp^{1/2}\big[\cals_n^*\hat\y - \cals_n\omega\big]\big\|_\calh\\
     &\leq 2\lambda^{-a}\big\|\calc_{n\lambda}^{1/2}\bar\calp^{1/2}(\bar\calp^{1/2}\calc_n\bar\calp^{1/2} +\lambda\mathcal{I})^{-1}\calp^{1/2}\calc_{n\lambda}^{1/2}\big\|\cdot\underbrace{\big\|\calc_{n\lambda}^{-1/2}\big[\cals_n^*\hat\y - \calc_n\omega\big]\big\|_\calh}_{\text{sample variance}}.
\end{align*}
The second equality holds because $\calz_n = \sqrt{n}\cals_n$ and $\cals_n^*\cals_n=\calc_n$ and $\hat\y = \y/\sqrt{n}$.
We use Lemma \ref{l:sample_variance} to upper bound the term $\big\|\calc_{n\lambda}^{-1/2}\big[\cals_n^*\hat\y - \calc_n\omega^\lambda\big]\big\|_\calh$. The remaining multiplicative term can be upper bounded by $1$ as 
\begin{align*}
    \bar\calp^{1/2}\calc_n\bar\calp^{1/2} +\lambda\mathcal{I} &\succeq  \bar\calp^{1/2}\calc_{n\lambda}\bar\calp^{1/2} = \bar\calp^{1/2}\calc_{n\lambda}^{1/2}\calc_{n\lambda}^{1/2}\bar\calp^{1/2}\\
    \Rightarrow (\bar\calp^{1/2}\calc_n\bar\calp^{1/2} +\lambda\mathcal{I})^{-1} &\preceq (\bar\calp^{1/2}\calc_{n\lambda}^{1/2}\calc_{n\lambda}^{1/2}\bar\calp^{1/2})^\dagger = (\calc_{n\lambda}^{1/2}\bar\calp^{1/2})^\dagger(\bar\calp^{1/2}\calc_{n\lambda}^{1/2})^\dagger.
\end{align*}
Therefore,
\begin{align*}
    \big\|\calc_{n\lambda}^{1/2}\bar\calp^{1/2}(\bar\calp^{1/2}\calc_n\bar\calp^{1/2} +\lambda\mathcal{I})^{-1}\bar\calp^{1/2}\calc_{n\lambda}^{1/2}\big\| \leq \|\calc_{n\lambda}^{1/2}\bar\calp^{1/2}(\calc_{n\lambda}^{1/2}\bar\calp^{1/2})^\dagger(\bar\calp^{1/2}\calc_{n\lambda}^{1/2})^\dagger\bar\calp^{1/2}\calc_{n\lambda}^{1/2}\| \leq 1.
\end{align*}
The upperbound on the variance term becomes
\begin{align}
    \mathcal{T}_v \leq \lambda^{-a}\cdot  O\Big(\lambda^\zeta + \frac{1}{n\lambda^{1-\zeta}} + \frac{1}{\sqrt{n\lambda^{\gamma}}} \Big)\log(1/\delta). \label{e:risk_4}
\end{align}
We now upper bound the bias term:
\begin{align*}
    \mathcal{T}_b&=\big\|\calc^{1/2-a}\big(\bar\calp\calz_n^*\bar\K_{n\lambda}^{-1}\calz_n- \mathcal{I}\big)\omega\big\|_\calh\\
    &=\big\|\calc^{1/2-a}\big(\bar\calp\calz_n^*\bar\K_{n\lambda}^{-1}\calz_n- \mathcal{I}\big)\calc_{\lambda}^{-1}\cals_\rho^*f_\calh\big\|_\calh\\
    &=\big\|\calc^{1/2-a}\big(\bar\calp\calz_n^*\bar\K_{n\lambda}^{-1}\calz_n- \mathcal{I}\big)\calc_{\lambda}^{-1}\cals_\rho^*\call^\zeta g\big\|_\calh.
\end{align*}
The last equality holds due to the source assumption \ref{as:4}. Using $\cals_\rho^*\call^\zeta = \cals_\rho^*(\cals_\rho\cals_\rho^*)^\zeta = (\cals_\rho^*\cals_\rho)^\zeta\cals_\rho^* = \calc^\zeta\cals_\rho^*g$, we get
\begin{align*}
    \mathcal{T}_b &= \big\|\calc^{1/2-a}\big(\bar\calp\calz_n^*\bar\K_{n\lambda}^{-1}\calz_n- \mathcal{I}\big)\calc_{\lambda}^{-1}\calc^\zeta\cals_\rho^*g \big\|_\calh\\
    &\leq  \big\|\calc^{1/2-a}\big(\bar\calp\calz_n^*\bar\K_{n\lambda}^{-1}\calz_n- \mathcal{I}\big)\calc_{\lambda}^{-1}\calc^\zeta\cals_\rho^*\big\|\cdot\big\|g\big\|_\rho\\
    &\leq R\big\|\calc^{1/2-a}\big(\bar\calp\calz_n^*\bar\K_{n\lambda}^{-1}\calz_n- \mathcal{I}\big)\calc_{\lambda}^{-1}\calc^\zeta\calc^{1/2}\big\|\\
    &\leq R\big\|\calc_\lambda^{1/2-a}\big(\bar\calp\calz_n^*\bar\K_{n\lambda}^{-1}\calz_n- \mathcal{I}\big)\calc_{\lambda}^{\zeta-1/2}\big\|.
\end{align*}
In the second to last inequality we used $\big\|g\big\|_\rho \leq R$ and in the last inequality we used $\calc^{\zeta+1/2} \preceq \calc_\lambda^{\zeta+1/2}$, and $\calc^{1/2-a} \preceq \calc_\lambda^{1/2-a}$ as $a \leq \zeta < \frac{1}{2}$. Continuing we get,
\begin{align*}
    \mathcal{T}_b &\leq R\big\|\calc_\lambda^{-a}\big\|\cdot\big\|\calc_\lambda^{1/2}\calc_{n\lambda}^{-1/2}\big\|\cdot\big\|\calc_{n\lambda}^{-1/2}\big\|\cdot\big\|\calc_{n\lambda}(\bar\calp\calz_n^*\bar\K_{n\lambda}^{-1}\calz_n-\mathcal{I})\big\|\cdot\|\calc_{\lambda}^{\zeta-1/2}\|\\
    &\leq \frac{2R\lambda^{-a}\lambda^\zeta}{\lambda}\big\|\calc_{n\lambda}(\bar\calp\calz_n^*\bar\K_{n\lambda}^{-1}\calz_n-\mathcal{I})\big\|\\
    &=\frac{2R\lambda^{-a}\lambda^\zeta}{n\lambda}\big\|(\calz_n^*\calz_n+n\lambda\I)(\bar\calp\calz_n^*\bar\K_{n\lambda}^{-1}\calz_n-\mathcal{I})\big\|\\
    &=\frac{2R\lambda^{-a}\lambda^\zeta}{n\lambda}\big\|(\calz_n^*\calz_n+n\lambda\I)(\calz_n^\dagger\calz_n\bar\calp\calz_n^*\bar\K_{n\lambda}^{-1}\calz_n-\mathcal{I})\big\|\\
    &=\frac{2R\lambda^{-a}\lambda^\zeta}{n\lambda}\big\|(\calz_n^*\calz_n+n\lambda\mathcal{I})(\calz_n^\dagger\bar\K\bar\K_{n\lambda}^{-1}\calz_n-\mathcal{I})\big\|,
\end{align*}
where we used the observation that $\calz_n^\dagger\calz_n$ is a projection on $\calh_n$, implying $\calz_n^\dagger\calz_n\bar\calp = \bar\calp$. Also we substituted $\calz_n\bar\calp\calz_n^* = \bar\K.$ Multiplying the terms we get,
\begin{align*}
    \mathcal{T}_b &\leq \frac{2R\lambda^{-a}\lambda^\zeta}{n\lambda}\big\|\calz_n^*\bar\K\bar\K_{n\lambda}^{-1}\calz_n + n\lambda\calz_n^\dagger\bar\K\bar\K_{n\lambda}^{-1}\calz_n-\calz_n^*\calz_n-n\lambda\mathcal{I}\big\|\\
    &\leq \frac{2R\lambda^{-a}\lambda^\zeta}{n\lambda}\Big(\|\calz_n^*(\bar\K\bar\K_{n\lambda}^{-1}-\I)\calz_n + n\lambda\calz_n^\dagger\bar\K\bar\K_{n\lambda}^{-1}\calz_n -n\lambda\mathcal{I}\|\Big).
\end{align*}
The terms $\big\|\calz_n^*(\bar\K\bar\K_{n\lambda}^{-1}-\I)\calz_n\big\|$ and $\big\|\calz_n^\dagger\bar\K\bar\K_{n\lambda}^{-1}\calz_n\big\|$ can be upper bounded separately as
\begin{align*}
    \big\|\calz_n^*(\bar\K\bar\K_{n\lambda}^{-1}-\I)\calz_n\big\| = n\lambda\|\calz_n^*\bar\K_{n\lambda}^{-1}\calz_n\| =n\lambda\|\K^{1/2}\bar\K_{n\lambda}^{-1}\K^{1/2}\| &\leq n\lambda\|\K_{n\lambda}^{1/2}\bar\K_{n\lambda}^{-1}\K_{n\lambda}^{1/2}\|
\end{align*}
where we used $\calz_n\calz_n^*=\K$. In second term substitute $\bar\K=\calz_n\bar\calp\calz_n^*$ to get:
\begin{align*}
    \|n\lambda\calz_n^\dagger\bar\K\bar\K_{n\lambda}^{-1}\calz_n\| &= n\lambda\|\calz_n^\dagger\calz_n\bar\calp\calz_n^*\bar\K_{n\lambda}^{-1}\calz_n\| \leq n\lambda\|\calz_n^\dagger\calz_n\|\cdot\|\bar\calp\|\cdot\|\calz_n^*\bar\K_{n\lambda}^{-1}\calz_n\|\\
    &\leq n\lambda\|\calz_n^*\bar\K_{n\lambda}^{-1}\calz_n\| \leq n\lambda\|\K_{n\lambda}^{1/2}\bar\K_{n\lambda}^{-1}\bar\K_{n\lambda}^{1/2}\|.
\end{align*}
Combining these upperbounds we get,
\begin{align}
    \mathcal{T}_b \leq 2R\lambda^{-a}\lambda^\zeta\big(1 + \|\K_{n\lambda}^{1/2}\bar\K_{n\lambda}^{-1}\K_{n\lambda}^{1/2}\|\big). \label{e:risk_5}
\end{align}
Combining (\ref{e:risk_5}), (\ref{e:risk_4}), (\ref{e:risk_3}), (\ref{e:risk_2}), and $(\ref{e:risk_1})$, we conclude
\begin{align*}
    \big\|\call^{-a}(\cals_\rho\omega_{[q]}- f_\calh)\big\|_\rho \leq   \lambda^{-a}\cdot O\Big(\lambda^\zeta\big\|\bar\K_{n\lambda}^{1/2}\bar\K_{n\lambda}^{-1}\bar\K_{n\lambda}^{1/2}\big\| + \frac{1}{n\lambda^{1-\zeta}} + \frac{1}{\sqrt{n\lambda^{\gamma}}} \Big)\log(1/\delta).
\end{align*}
\end{proof}
Using Theorem \ref{t:main_stat_risk_2} we prove the following two corollaries for classical Nystr\"om and Block Nystr\"om KRR respectively, recovering the optimal generalization error up to a multiplicative factor of $\alpha$.

For the following corollary, consider $\lambda' =\alpha\lambda$ and construct a Nystr\"om approximation to $\K$ by sampling landmarks from $n\lambda'$-ridge leverage scores of $\K$, but still adding $n\lambda$ as the regularizer. We refer to this as classical Nystr\"om.
\begin{corollary}[Statistical risk of KRR with classical Nystr\"om]\label{c:single_nys_risk}
      Let $0\leq a \leq \zeta$, $\zeta < \frac{1}{2}$, $n \geq O(G^2\log(G^2/\delta))$, $\lambda^* = \tilde O(n^{-\frac{1}{2\zeta+\gamma}})$ if $2\zeta+\gamma >1$, and $\lambda^* = \tilde O(n^{-1})$ if $2\zeta+\gamma \leq 1$. Let $\lambda' = \alpha\lambda^* < 1$ for some $\alpha>1$. Then under assumptions \ref{as:1}-\ref{as:4} and for any $\delta>0$
      \begin{align*}
    \big\|\call^{-a}(\cals_\rho\hat\omega- f_\calh)\big\|_\rho \leq \begin{cases} \alpha\cdot \tilde O(n^{-\frac{\zeta-a}{2\zeta+\gamma}})  \ \text{if} \ \ 2\zeta+\gamma>1 \\
    \alpha\cdot \tilde O(n^{-(\zeta-a)}) \ \ \text{if} \ 2\zeta+\gamma \leq 1
\end{cases}
\end{align*}
with probability $1-\delta$. Here $\hat\omega$ denotes the hypothesis $\hat\calp\calz_n^*(\hat\K+n\lambda^*\I)^{-1}\y$, $\hat\calp$ is orthogonal projection onto the $m$ dimensional subspace spanned by Nystr\"om landmarks sampled using $O(1)$-approximate $n\lambda'$-ridge leverage scores of $\K$ and $\hat\K$ denotes the corresponding Nystr\"om approximation to $\K$.
\end{corollary}

\begin{proof}
    Let $\lambda^*$ be the optimal regularizer depending on the case if $2\zeta+\gamma>1$ or $2\zeta+\gamma \leq 1$. Consider $\lambda' = \alpha\lambda^*$. We have with probability $1-\delta$
    \begin{align*}
        \K+n\lambda^*\I \preceq \K+n\lambda'\I \preceq 2(\hat\K +n\lambda'\I) \preceq \frac{2\lambda'}{\lambda^*}(\hat\K+n\lambda^*\I) \preceq 2\alpha(\K+n\lambda^*\I)
    \end{align*}
    implying $\|(\K+n\lambda^*\I)^{1/2}\big(\hat\K+n\lambda^*\I\big)^{-1}(\K+n\lambda^*\I)^{1/2}\| \leq 2\alpha$. Substituting this in the statement of Theorem \ref{t:main_stat_risk_2} finishes the proof.
\end{proof}
We provide an extended version of Corollary \ref{c:avg_nys_risk}. The following extension is useful when the excess risk factor $\alpha$ is potentially quite large and the resulting linear system $(\hat\K_{[q]} +n\lambda^{*}\I)\w=\y$ is more expensive to solve than sampling Nystr\"om landmarks from $\alpha^2\lambda^*$-ridge leverage scores of $\K$, due to large number of blocks as $q=\tilde O(\alpha)$. To recover corollary \ref{c:avg_nys_risk} one can consider $\beta=\alpha$ in the following result.
\begin{corollary}[Expected risk of Block-Nystr\"om-KRR]\label{c:avg_nys_risk_2}
      Under Assumptions \ref{as:1}, \ref{as:2}, \ref{as:3} and \ref{as:4}, suppose that $\zeta<1/2$ and $n \geq \tilde O(G^2\log(G^2/\delta))$. Also, let $\lambda^* = \tilde O(n^{-\frac{1}{2\zeta+\gamma}})$ if $2\zeta+\gamma >1$ and $\lambda^* = \tilde O(n^{-1})$ if $2\zeta+\gamma \leq 1$. Consider Block-Nystr\"om-KRR, $\hat\omega_{[q]}=\calp_{[q]}\calz_n^*(\hat\K_{[q]}+n\lambda^*\I)^{-1}\y$, constructed using $q=\tilde O(\beta)$ blocks, each with $\tilde O(d_{n\lambda'}(\K))$ leverage score samples where $1\leq \beta \leq \alpha$ and $\lambda'=\beta\alpha\lambda^*$. Then,
      \begin{align*}
    \big\|\cals_\rho\hat\omega_{[q]}- f_\calh\big\|_\rho \leq \begin{cases} \alpha\cdot \tilde O(n^{-\frac{\zeta}{2\zeta+\gamma}})  &\text{if} \ \ 2\zeta+\gamma>1,\\
    \alpha\cdot \tilde O(n^{-\zeta}) &\text{if} \  \ 2\zeta+\gamma \leq 1.
\end{cases}
\end{align*}
\end{corollary}
\begin{proof}\textbf{of Corollary \ref{c:avg_nys_risk_2}.}
    Let $\lambda^*$ be the optimal regularizer depending on the case if $2\zeta+\gamma>1$ or $2\zeta+\gamma \leq 1$. Consider $\tilde\lambda= \frac{\alpha}{\beta}\lambda^*$ and $\lambda' = \beta\alpha \lambda^*$. We have with probability $1-\delta$,
    \begin{align*}
        \K+n\lambda^*\I \preceq \K+n\tilde\lambda\I \preceq 8 \sqrt{\frac{\lambda'}{\tilde\lambda}}\Big(\frac{1}{q}\sum_{i=1}^{q}{\hat\K_i +n\tilde\lambda\I}\Big) \preceq 8 \sqrt{\frac{\lambda'}{\tilde\lambda}}\cdot\frac{\tilde\lambda}{\lambda^*}\Big(\frac{1}{q}\sum_{i=1}^{q}{\hat\K_i +n\lambda^*\I}\Big) \preceq 8\alpha(\K+n\lambda^*\I)
    \end{align*}
    For $q > O(\beta\log(n/\delta))$. In the second inequality, we used Theorem \ref{t:Nys_spectral_approximation_2}. Therefore,
    \begin{align*}
       (\K+n\lambda^*\I)^{-1} \preceq \Big(\frac{1}{q}\sum_{i=1}^{q}{\hat\K_i +n\lambda^*\I}\Big)^{-1} \preceq 8\alpha(\K+n\lambda^*\I)^{-1}
    \end{align*}
    implying $\Big\|(\K+n\lambda^*\I)^{1/2}\Big(\frac{1}{q}\sum_{i=1}^{q}{\hat\K_i +n\lambda^*\I}\Big)^{-1}(\K+n\lambda^*\I)^{1/2}\Big\| < 8\alpha$. Substituting this in the statement of Theorem \ref{t:main_stat_risk_2} finishes the proof.
\end{proof}
In the following lemma, we derive explicit expressions for the time complexity of classical Nystr\"om KRR for any $\alpha \geq 1$. In particular, we sample landmarks using $n\lambda'$-ridge leverage scores of $\K$ where $\lambda'=\alpha\lambda^*$, obtaining the Nystr\"om approximation $\hat\K$. We then add $n\lambda^*\I$ to $\hat \K$, obtaining an $n\lambda^*$-regularized $\alpha$-approximation to $\K$.
For better intuition, it is helpful to think of $\alpha$ as $n^\theta$ for small $\theta>0$. In particular, the form of the exponent $\theta$ changes as the regime changes from $2\zeta+\gamma>1$ to $2\zeta+\gamma \leq 1$. We capture this regime change by considering $\theta= \mu/\max\{1,2\zeta+\gamma\}$. We consider $a=0$ to obtain the time complexity for obtaining $\alpha$ multiplicative bound for the expected risk.
\begin{lemma}[Cost analysis of KRR with classical Nystr\"om]\label{l:cost_single_Nystrom}
Let $a=0$ and $\zeta < \frac{1}{2}$. Let $\lambda' = \alpha\lambda^*$, where $\lambda^* = \tilde O(n^{-\frac{1}{2\zeta+\gamma}})$ and $\alpha = n^{\frac{\mu}{2\zeta+\gamma}}$ if $2\zeta+\gamma >1$, and $\lambda^* = \tilde O(n^{-1})$ and $\alpha = n^{\mu}$ if $2\zeta+\gamma <1$. Furthermore, let $\mu \leq \zeta$. Then the total time complexity of classical Nystr\"om KRR in regime of $2\zeta+\gamma >1$ is given as:
 \begin{align*} \mathcal{T}_{nys}=\begin{cases} 
        \tilde O\Big(n^{\frac{(1+2\gamma)(1-\mu)}{2\zeta+\gamma}}\Big)
         \ \text{ if  $\mu \leq \frac{1-2\zeta}{1+\gamma}$}\\
        \tilde O\Big(n^{\frac{2\zeta +\gamma(2-\mu)}{2\zeta+\gamma}}\Big) \ \ \ \text{ otherwise}.
    \end{cases}
    \end{align*}
and if $2\zeta+\gamma \leq 1$,
      \begin{align*}
       \mathcal{T}_{nys}=
    \begin{cases}
       \tilde O\Big(n^{(1+2\gamma)(1-\mu)}\Big) \text{if} \ \mu \leq \frac{\gamma}{1+\gamma},\\
        \tilde O\Big(n^{1+\gamma(1-\mu)}\Big) \text {otherwise.} 
    \end{cases}
    \end{align*}
\end{lemma}
\begin{proof}
    In the regime of $2\zeta+\gamma >1$, the optimal $\lambda^*$ is $\tilde O(n^{-\frac{1}{2\zeta+\gamma}})$. Let $\lambda' =\alpha\lambda^*$ where $\alpha = n^{\frac{\mu}{2\zeta+\gamma}}$. Let $m=d_{n\lambda'}(\K)$.
    The overall time complexity of classical Nystr\"om KRR is the sum of three costs as follows:
    \begin{itemize}
        \item Landmarks sampling using Lemma \ref{l:rls_bless} : $\tilde O(m^2/\lambda') = \tilde O(\frac{1}{(\lambda')^{1+2\gamma}}\big) = \tilde O(\frac{1}{\alpha^{1+2\gamma}}\cdot n^{\frac{1+2\gamma}{2\zeta+\gamma}}).$
        \item Cost of inverting $\tilde O(m)\times \tilde O(m)$  submatrix of $\K$: $\tilde O(m^3) = \tilde O(\frac{1}{\alpha^{3\gamma}}\cdot n^{\frac{3\gamma}{2\zeta+\gamma}}).$
        \item Cost of solving linear system $(\hat\K+n\lambda^*\I)\w=\y$  using preconditioned conjugate gradient: $ \tilde O(nm) = \tilde O(\frac{1}{\alpha^\gamma}\cdot n^{\frac{2\zeta+2\gamma}{2\zeta+\gamma}})$.
    \end{itemize}
    Under the assumption that $\mu \leq\zeta$, it is easy to show that cost of sampling landmarks is larger than the cost of inverting the principal submatrix. Now comparing the first and third costs it is easy to see that cost of landmarks sampling dominates the cost of solving the linear system as long as $\mu < \frac{1-2\zeta}{1+\gamma}.$
    Therefore overall time complexity is given as
    \begin{align*}\begin{cases} 
        \tilde O\Big(n^{\frac{(1+2\gamma)(1-\mu)}{2\zeta+\gamma}}\Big)
         \ \text{ if  $\mu \leq \frac{1-2\zeta}{1+\gamma}$}\\
        \tilde O\Big(n^{\frac{2\zeta +\gamma(2-\mu)}{2\zeta+\gamma}}\Big) \ \ \ \text{ otherwise}.
    \end{cases}
    \end{align*}
    Note that we always assume $\mu \leq \zeta$.
    Now consider the regime $2\zeta+\gamma <1$, the optimal $\lambda^* = \tilde O(1/n)$ and the corresponding learning rate $\tilde O(n^{-\zeta})$. Let $\lambda' = \alpha\lambda^*$ where now $\alpha = n^{\mu}$ for $0 \leq \mu \leq \zeta$. We have
    \begin{itemize}
        \item Nystr\"om landmarks sampling using Lemma \ref{l:rls_bless}: $\tilde O(m^2/\lambda') = \tilde O\Big(n^{(1+2\gamma)(1-\mu)}\Big).$
        \item  Cost of inverting $\tilde O(m)\times \tilde O(m)$ matrix : $\tilde O(m^3) = \tilde O(n^{3\gamma(1-\mu)})$
        \item Cost of solving linear system $(\hat\K+n\lambda^*\I)\w=\y$  using preconditioned conjugate gradient: $\tilde O(n^{1+\gamma(1-\mu)})$
    \end{itemize}
    The overall cost is then given as
    \begin{align*}
    \begin{cases}
        \tilde O\Big(n^{(1+2\gamma)(1-\mu)}\Big) \text{if} \ \mu \leq \frac{\gamma}{1+\gamma},\\
        \tilde O\Big(n^{1+\gamma(1-\mu)}\Big) \text {otherwise.} 
    \end{cases}
    \end{align*}
    Here also additional constraint is $\mu \leq \zeta$.
\end{proof}
We now provide explicit expressions for the time complexity of Block-Nystr\"om KRR. 
\begin{lemma}[Cost analysis of KRR with Block-Nystr\"om]\label{l:cost_avg_Nystrom}
Let $a=0$, $\zeta<\frac{1}{2}$ and $\lambda^*$ be the optimal regularizer depending on the regime if $2\zeta+\gamma>1$ or $2\zeta+\gamma\leq 1$. Let $\lambda' = \beta\alpha\lambda^*$ for some $1\leq \beta\leq \alpha$ and define $q$ as follows
    \begin{align*}
        q=\begin{cases}
            \lceil200\beta\log(n)/\delta\rceil \ \text{if} \ \beta>1\\
            1 \ \text{if} \ \beta=1.
        \end{cases}
    \end{align*}
If $2\zeta+\gamma>1$, let $\lambda^* = \tilde O(n^{-\frac{1}{2\zeta+\gamma}})$, $\alpha = n^{\frac{\mu}{2\zeta+\gamma}}$ and $\beta=n^{\frac{\nu}{2\zeta+\gamma}}$  where $0\leq \mu \leq \zeta$ and $\nu = \max\Big\{0,\min\big\{\mu,\frac{1-2\zeta-\mu(1.5+\gamma)}{\gamma+\phi+0.5}\big\}\Big\}$. 
Then the total time complexity of Block-Nystr\"om KRR is given as:
   \begin{align*}\mathcal{T}_{Blk}=
        \begin{cases}
            \tilde O\Big(n^{\frac{(1+2\gamma)(1-\mu)-2\gamma\nu}{2\zeta+\gamma}}\Big) \ \text{if} \ \mu < \frac{1-2\zeta}{\gamma+1.5}, \\
         \tilde O\Big(n^{\frac{(1+2\gamma)(1-\mu)}{2\zeta+\gamma}}\Big) \ \ \ \text{if} \ \frac{1-2\zeta}{\gamma+1.5} \leq \mu \leq \frac{1-2\zeta}{\gamma+1},\\
         \tilde O\Big(n^{\frac{2\zeta+\gamma(2-\mu)}{2\zeta+\gamma}}\Big) \text{otherwise}.
        \end{cases}
    \end{align*}
    If $2\zeta+\gamma \leq 1$, let $\lambda^* = \tilde O(n^{-1})$, $\alpha = n^{\mu}$, $\beta=n^{\nu}$ where $0 \leq \mu\leq \zeta$ and  $\nu = \max\Big\{0,\min\big\{\mu, \frac{\gamma-\mu(\gamma+1.5)}{\gamma+\phi+0.5}\big\}\Big\}$. Then the overall cost of Block-Nystr\"om is given as
    \begin{align*}
    \mathcal{T}_{Blk}=
        \begin{cases}
            \tilde O\Big(n^{(1+2\gamma)(1-\mu)-2\gamma\nu}\Big) \ \text{if} \ \mu < \frac{\gamma}{\gamma+1.5} \\
         \tilde O\Big(n^{(1+2\gamma)(1-\mu)}\Big) \ \ \ \text{if} \ \frac{\gamma}{\gamma+1.5} \leq \mu \leq \frac{\gamma}{1+\gamma}\\
         \tilde O\Big(n^{1+\gamma(1-\mu)}\Big) \text{otherwise}
        \end{cases}
    \end{align*}
\end{lemma}
\begin{proof}
We consider $\tilde\lambda= \frac{\alpha}{\beta}\lambda^*$ and $\lambda'=\beta\alpha\lambda^*$. For sampling Nystr\"om landmarks we use Lemma \ref{l:rls_bless} $q$ times leading to cost of $\tilde O(q\cdot m^2/\lambda')$, where $m=d_{n\lambda'}(\K)$. First consider the regime $2\zeta+\gamma>1$. Similar to Lemma \ref{l:cost_single_Nystrom} we break down the cost in three components:
    \begin{itemize}
        \item Nystr\"om landmarks sampling : $\tilde O(q m^2/\lambda') = \tilde O(\beta^{-2\gamma}\alpha^{-(1+2\gamma)} n^{\frac{1+2\gamma}{2\zeta+\gamma}}).$
     \item  Cost of inverting $q$ $\tilde O(m)\times \tilde O(m)$ principal submatrices of $\K$: $\tilde O(qm^3) =\tilde O\Big(\beta^{1-3\gamma}\alpha^{-3\gamma} n^{\frac{3\gamma}{2\zeta+\gamma}}\Big).$
          \item Cost of solving linear system $(\hat\K_{[q]}+n\lambda^*\I)\w=\y$ using techniques from Lemma \ref{t:linear_system_solve_2}: $\tilde O\Big(\sqrt{\tilde\lambda/\lambda^*}\cdot q^{1+\phi}nm\Big) = \tilde O\Big(\sqrt{\alpha/\beta}\cdot\beta^{1+\phi-\gamma}\alpha^{-\gamma}n^{\frac{2\zeta+2\gamma}{2\zeta+\gamma}}\Big) = \tilde O \Big(\beta^{1/2+\phi-\gamma}\alpha^{1/2-\gamma}n^{\frac{2\zeta+2\gamma}{2\zeta+\gamma}}\Big).$
    \end{itemize}
 Under the assumption that $\mu \leq \zeta$, one can show that the cost of landmarks sampling is larger than the cost of inverting the principal submatrices. Comparing the first and the third costs we get
    \begin{align*}
        &\beta^{-2\gamma}\alpha^{-(1+2\gamma)} n^{\frac{1+2\gamma}{2\zeta+\gamma}} \geq \beta^{1/2+\phi-\gamma}\alpha^{1/2-\gamma}n^{\frac{2\zeta+2\gamma}{2\zeta+\gamma}}\\
        \Rightarrow&  n^{\frac{1-2\zeta}{2\zeta+\gamma}} \geq \beta^{1/2+\phi+\gamma}\alpha^{3/2+\gamma}\\
         \Rightarrow& 1-2\zeta \geq \mu(1.5+\gamma) + \nu(0.5+\phi+\gamma)\\
         \Rightarrow& \nu \leq \frac{1-2\zeta-\mu(1.5+\gamma)}{\gamma+\phi+0.5}.
    \end{align*}
     We set $\nu = \max\Big\{0,\min\big\{\mu,\frac{1-2\zeta-\mu(1.5+\gamma)}{\gamma+\phi+0.5}\big\}\Big\}$. If $\nu >0$, then we run  Block-Nystr\"om KRR and if $\nu=0$ i.e. $\beta=1$ and consequently $q=1$, we run classical Nystr\"om KRR. The overall time complexity is given as
 \begin{align*}
        \begin{cases}
            \tilde O\Big(n^{\frac{(1+2\gamma)(1-\mu)-2\gamma\nu}{2\zeta+\gamma}}\Big) \ \text{if} \ \mu < \frac{1-2\zeta}{\gamma+1.5}, \\
         \tilde O\Big(n^{\frac{(1+2\gamma)(1-\mu)}{2\zeta+\gamma}}\Big) \ \ \ \text{if} \ \frac{1-2\zeta}{\gamma+1.5} \leq \mu \leq \frac{1-2\zeta}{\gamma+1},\\
         \tilde O\Big(n^{\frac{2\zeta+\gamma(2-\mu)}{2\zeta+\gamma}}\Big) \ \text{otherwise}.
        \end{cases}
    \end{align*}
    Now let $2\zeta+\gamma \leq 1$. The cost of  Block-Nystr\"om KRR is given as
   \begin{itemize}
       \item Landmarks sampling: $\tilde O(qm^2/\lambda') = \tilde O\Big(\beta^{-2\gamma}\alpha^{-(1+2\gamma)}n^{1+2\gamma}\Big).$
        \item  Cost of inverting $q$ $\tilde O(m)\times \tilde O(m)$ principal submatrices of $\K$: $\tilde O(qm^3) : \tilde O\Big(\beta^{1-3\gamma}\alpha^{-3\gamma}n^{3\gamma}\Big)$
        \item Cost of solving linear system $(\hat\K_{[q]}+n\lambda^*\I)\w=\y$ using techniques from Lemma \ref{t:linear_system_solve_2}: $\tilde O\Big(\sqrt{\tilde\lambda/\lambda^*}\cdot q^{1+\phi}nm\Big) = \tilde O\Big(\sqrt{\alpha/\beta}\cdot\beta^{1+\phi-\gamma}\alpha^{-\gamma}n^{1+\gamma}\Big) = \tilde O\Big(\beta^{1/2+\phi-\gamma}\alpha^{1/2-\gamma}n^{1+\gamma}\Big)$
   \end{itemize}
    Again, one can show that the cost of landmarks sampling is larger than the cost of inverting the principal submatrices. Comparing the first and third costs we get
    \begin{align*}
        \beta^{-2\gamma}\alpha^{-(1+2\gamma)}n^{1+2\gamma} &\geq \beta^{1/2+\phi-\gamma}\alpha^{1/2-\gamma}n^{1+\gamma}\\
        \Rightarrow n^\gamma &\geq \beta^{1/2+\phi+\gamma}\alpha^{3/2+\gamma}\\
        \Rightarrow \gamma &\geq \mu(\gamma+3/2) + \nu(\gamma+\phi+1/2)\\
        \Rightarrow \nu &\leq \frac{\gamma-\mu(\gamma+1.5)}{\gamma+\phi+0.5}
    \end{align*}
    As $0\leq\nu \leq \mu$ we set $\nu = \max\Big\{0,\min\big\{\mu, \frac{\gamma-\mu(\gamma+1.5)}{\gamma+\phi+0.5}\big\}\Big\}$. Again, $\nu >0$ corresponds to Block-Nystr\"om and if $\nu=0$ we run classical Nystr\"om KRR. The overall cost is therefore given as
    \begin{align*}
        \begin{cases}
            \tilde O\Big(n^{(1+2\gamma)(1-\mu)-2\gamma\nu}\Big) \ \text{if} \ \mu < \frac{\gamma}{\gamma+1.5} \\
         \tilde O\Big(n^{(1+2\gamma)(1-\mu)}\Big) \ \ \ \text{if} \ \frac{\gamma}{\gamma+1.5} \leq \mu \leq \frac{\gamma}{1+\gamma}\\
         \tilde O\Big(n^{1+\gamma(1-\mu)}\Big) \ \ \text{otherwise}.
        \end{cases}
    \end{align*}
\end{proof}
\paragraph{Detailed cost comparison between classical Nystr\"om and Block Nystr\"om KRR.}
As shown in Lemma \ref{l:cost_avg_Nystrom} there is a regime in which Block-Nystr\"om provides computational gains over classical Nystr\"om KRR. In particular, we have proven that as long as $\alpha=n^\theta$ for $\theta < \frac{1}{\gamma+1.5}\cdot\frac{\min\{\gamma,1-2\zeta\}}{\max\{1,2\zeta+\gamma\}}$, Block-Nystr\"om enjoys faster time complexity and achieves the same excess risk factor of $\alpha$ as obtained by classical Nystr\"om. Furthermore, within this range of values for $\theta$, the computational gains enjoyed by Block-Nystr\"om change their trend. As $\theta$ increases from $0$ to $\frac{1}{2\gamma+2+o(1)}\cdot\frac{\min\{\gamma,1-2\zeta\}}{\max\{1,2\zeta+\gamma\}}$, the ratio of time complexity of two methods, $\mathcal{T}_{Blk}/\mathcal{T}_{nys}$, scales as $1/n^{2\theta\gamma}$, whereas when $\theta$ increases from $\frac{1}{2\gamma+2+o(1)}\cdot\frac{\min\{\gamma,1-2\zeta\}}{\max\{1,2\zeta+\gamma\}}$ to $\frac{1}{\gamma+1.5}\cdot\frac{\min\{\gamma,1-2\zeta\}}{\max\{1,2\zeta+\gamma\}}$, we use a slightly different averaging scheme where we reduce the number of blocks to avoid computational overhead resulting due to potentially having to solve the linear system $(\hat\K_{[q]}+n\lambda^*\I)\w=\y$ with $q=\tilde O(\alpha)$ number of blocks. In this case we consider $q=\tilde O(\beta)$ with $\beta$ diminishing from $\alpha$ to $1$ as $\theta$ increases in this later regime. This idea of reduced averaging leads $\mathcal{T}_{Blk}/\mathcal{T}_{nys}$ to scale as $\frac{1}{n^{\xi}}$ where $\xi=\frac{1}{\gamma+o(1)+0.5}\Big(\frac{\min\{\gamma,1-2\zeta\}}{\max\{1,2\zeta+\gamma\}}-(\gamma+1.5)\theta\Big)$, meaning larger $\theta$ has diminishing computational gains. At the critical point of $\theta = \frac{1}{\gamma+1.5}\frac{\min\{\gamma,1-2\zeta\}}{\max\{1,2\zeta+\gamma\}}$, Block Nystr\"om KRR reduces to classical Nystr\"om KRR as $\beta$ becomes $1$. 
\section{Proof of Lemma \ref{l:rls-square-main}}
\label{a:preconditioner}
In this section we provide proof of Lemma \ref{l:rls-square-main}. We start by restating the lemma, replacing the notation $\bar\lambda$ with $\lambda$ for notational convenience. For psd matrices, we use $\A\approx_c\B$ to denote $c^{-1}\B\preceq\A\preceq c\B$, and $\omega \approx 2.372$ denotes the fast matrix multiplication exponent.

\begin{lemma}\label{l:rls-square}
    Given an $n\times n$ psd matrix $\A$ with at most $k$ eigenvalues larger than $O(1)$ times its smallest eigenvalue, consider $\lambda=\frac1k\sum_{i>k}\lambda_i(\A)$. Then, $d_\lambda(\A)\leq 2k$, and we can compute $O(1)$-approximations of all $\lambda$-ridge leverage scores of $\A$ in time $\tilde O(\nnz(\A) + k^\omega)$.
\end{lemma}

\begin{proof}
    The fact that $d_\lambda(\A)\leq 2k$ is shown in Lemma 2.1 of \cite{dmy24}. Next, let $\lambda_1\geq\lambda_2\geq ...\geq\lambda_n$ be the eigenvalues of $\A$. We will show that it suffices to approximate the ridge leverage scores of $\A^2$ instead of $\A$. Let $\tilde\lambda = \frac1k\sum_{i>k}\lambda_i^2$, so that $d_{\tilde\lambda}(\A^2)\leq 2k$. Also, define $\bar\kappa_k(\A)=\frac1{n-k}\sum_{i>k}\frac{\lambda_i}{\lambda_n}$. We're going to show that:
   \begin{align}
     \bar\kappa_k(\A)\cdot \A(\A+\lambda\I)^{-1}\preceq \bar\kappa_k(\A^2)\cdot \A^2(\A^2+\tilde\lambda\I)^{-1}.
     \label{eq:lev-1}
   \end{align}
    Since the matrices commute, it is enough to show this for each
eigenvalue. Observe that we have:
\begin{align*}
  \frac{\bar\kappa_k(\A^2)}{\bar\kappa_k(\A)} =
  \frac{\sum_{j>k}\lambda_j^2}{\sum_{j>k}\lambda_n\lambda_j}\in [1,\lambda_{k+1}/\lambda_n].
\end{align*}
This means that
\begin{align*}
  \frac{\lambda_i}{\lambda_i+\frac1k\sum_{j>k}\lambda_j}
  &=\frac{\lambda_i^2}{\lambda_i^2+\frac1k\sum_{j>k}\lambda_i\lambda_j}
    \leq\frac{\lambda_i^2}{\lambda_i^2+\frac1k\sum_{j>k}\lambda_n\lambda_j}
  \\
  &=\frac{\lambda_i^2}{\lambda_i^2+\frac{\bar\kappa_k(\A)}{\bar\kappa_k(\A^2)}\frac1k\sum_{j>k}\lambda_j^2}
  \leq \frac{\bar\kappa_k(\A^2)}{\bar\kappa_k(\A)}\frac{\lambda_i^2}{\lambda_i^2+\frac1k\sum_{j>k}\lambda_j^2},
\end{align*}
which implies \eqref{eq:lev-1}. Since $\lambda$-ridge leverage scores of $\A$ are the diagonal entries of $\A(\A+\lambda\I)^{-1}$, this shows that for any $i$:
\begin{align*}
    \ell_{i}(\A,\lambda) 
    &= \e_i^\top\A(\A+\lambda\I)^{-1}\e_i
    \\
    &\leq \frac{\bar\kappa_k(\A^2)}{\bar\kappa_k(\A)}\e_i^\top\A^2(\A^2+\tilde\lambda\I)^{-1}\e_i
    \\
    &=O(1)\cdot \ell_{i}(\A^2,\tilde\lambda),
\end{align*}
where in the last step we used that $\bar\kappa_k(\A^2) = O(1)$ because all of the tail eigenvalues are within a constant of each other. This, together with the fact that $\sum_i\ell_{i}(\A^2,\tilde\lambda) = O(k)$, implies that it suffices to perform approximate ridge leverage score sampling with respect to $\ell_{i}(\A^2,\tilde\lambda)$. 

Next, we show how to construct approximations of $\ell_{i}(\A^2,\tilde\lambda)$ in $\tilde O(\nnz(\A)+k^\omega)$ time. First, notice that
\begin{align*}
    \ell_{i}(\A^2,\tilde\lambda) = \a_i^\top(\A^2+\tilde\lambda\I)^{-1}\a_i,
\end{align*}
where $\a_i$ is the $i$th row of $\A$. Next, we approximate the inner matrix by using a $\tilde O(k)\times n$ sparse sketching matrix $\Pi$ \citep{cohen2016optimal}, computing $\M=\Pi\A$ in $\tilde O(\nnz(\A))$ time, so that with high probability $\M^\top\M+\tilde\lambda\I\approx_2 \A^2+\tilde\lambda\I$. Thus, it suffices to approximate:
\begin{align*}
\a_i^\top(\M^\top\M+\tilde\lambda\I)^{-1}\a_i
&= \|(\M^\top\M+\tilde\lambda\I)^{-1/2}\a_i\|^2
\\
&= \|\B_{\tilde\lambda}(\M^\top\M+\tilde\lambda\I)^{-1}\a_i\|^2,\qquad\text{where}
    \quad\B_{\tilde\lambda} = \begin{bmatrix}\M\\\sqrt{\tilde\lambda}\I\end{bmatrix}.
\end{align*}
Note that if $\S$ is an $O(\log n)\times 2n$ Johnson-Lindenstrauss embedding, then we can approximate the above norms to within a constant factor with $\tilde\ell_i = \|\S\B_{\tilde\lambda}(\M^\top\M+\tilde\lambda\I)^{-1}\a_i\|^2$. In order to compute these estimates, we first pre-compute the $O(\log n)\times d$ matrix $\S\B_{\tilde\lambda}$, and then compute $\mathbf{R}=\S\B_{\tilde\lambda}\cdot(\M^\top\M+\tilde\lambda\I)^{-1}$ by solving $O(\log n)$ linear systems with $\M^\top\M+\tilde\lambda\I$. The linear systems can be solved by following the preconditioning strategy of \cite{dmy24}, observing that: 
  \begin{align*}
    (\M^\top\M+\tilde\lambda\I)^{-1}
    = \frac1{\tilde\lambda}\Big(\I -
    \M^\top(\M\M^\top+\tilde\lambda\I)^{-1}\M\Big).
  \end{align*}
Thus, it suffices to solve a linear system with $\M\M^\top+\tilde\lambda\I$. To do this, we construct a $k\times \tilde O(k)$ sketch $\tilde\M = \M\S_2^\top$ such that $\tilde\M\tilde\M^\top \approx_{O(1)} \M\M^\top$ via a subspace embedding \citep{cohen2016optimal}, and precondition that linear system with $\B^{-1}=(\tilde\M\tilde\M^\top+\tilde\lambda\I)^{-1}\approx_{O(1)}(\M\M^\top+\tilde\lambda\I)^{-1}$. Thus, after precomputing
$\B^{-1}$ at the cost of
$\tilde O(\nnz(\M)+k^\omega)$, we can solve
the linear system $(\M^\top\M+\tilde\lambda\I)^{-1}\v$ in time $\tilde O(\nnz(\A)+k^2)$.

Finally, we compute the $O(\log n)\times n$ matrix $\mathbf{R}\A$ and then let $\tilde\ell_i$ be the squared row norms of this matrix. The overall costs are $\tilde O(\nnz(\A)+k^\omega)$ for precomputing $\mathbf{R}$ and then $\tilde O(\nnz(\A))$ for computing the squared row norms.

\end{proof}

\end{document}